\documentclass[12pt,reqno,pre,longbibliography]{amsart}
\usepackage{amsfonts}
\usepackage{graphicx}
\usepackage{amscd}
\usepackage{amsmath,amssymb}
\usepackage{epsfig}
\usepackage{comment}
\usepackage[usenames]{color}
\usepackage{subfigure}
\usepackage[english]{babel}

\providecommand{\U}[1]{\protect\rule{.1in}{.1in}}
\providecommand{\U}[1]{\protect\rule{.1in}{.1in}}
\allowdisplaybreaks
\newtheorem{theorem}{Theorem}

\theoremstyle{plain}
\newtheorem{acknowledgement}{Acknowledgement}

\numberwithin{equation}{section}

\DeclareMathOperator{\sech}{sech}

\DeclareMathOperator{\gd}{gd}

\begin{document}
\title[Generalized Coupled NLS System]{On explicit soliton solutions and blow-up for
Coupled variable coefficient nonlinear Schr\"{o}dinger equations}
\author{José M. Escorcia }
\address{Escuela de Ciencias Aplicadas e Ingenier\'ia, Universidad EAFIT, Carrera 49 No. 7 Sur-50, Medellín 050022, Colombia.}
\email{jmescorcit@eafit.edu.co}
\author{Erwin Suazo}
\address{School of Mathematical and Statistical Sciences, University of
Texas Rio Grande Valley, 1201 W. University Drive, Edinburg, Texas,
78539-2999.}
\email{erwin.suazo@utrgv.edu}
\date{\today }


\begin{abstract}
This work is concerned with the study of explicit solutions for a generalized coupled nonlinear Schr\"{o}dinger equations (NLS) system with variable coefficients. Indeed, we show, employing similarity transformations, the existence of Rogue wave and dark-bright soliton like-solutions for such a generalized NLS system, provided the coefficients satisfy a Riccati system. As a result of the multiparameter solution of the Riccati system, the nonlinear dynamics of the solution can be controlled. Finite-time singular solutions in the $L^{\infty}$ norm  for the generalized coupled NLS system are presented explicitly. Finally, an n-dimensional transformation between a variable coefficient NLS coupled system and a constant coupled system coefficient is presented. Soliton and Rogue wave solutions for this high-dimensional system are presented as well. A Mathematica file has been prepared as supplementary
material, verifying the Riccati systems used in the construction of the solutions \\


\end{abstract}

\maketitle

\section{Introduction}


The coupled nonlinear Schr$\ddot{\mbox{o}}$dinger equation models several real phenomena such as the interaction of Bloch-wave packets in a periodic system \cite{Shi}, the evolution of two surface wave packets in deep water, and in wavelength-division-multiplexed systems \cite{Menyuk}. The starting point of many studies of these kinds of models was  the well-known  Manakov system  introduced in 1974  \cite{Manakov1974} (see Equations (\ref{ManakovSystem1})-(\ref{ManakovSystem2})). The Manakov system  is an integrable coupled NLS  characterized by equal nonlinear interactions within and between  components \cite{ZAKHAROV1982}.
A great effort has been made  to figure out  the dynamics of the soliton-like solutions for coupled NLS. Some of these structures are  dark-bright (DB), bright-bright (BB), and dark-dark (DD) solitons. These states have been observed in several experiments in optics \cite{Zhigang1997,Ostrovskaya1999} and in the context of Bose-Einstein condensates \cite{Busch2001}. Another important class of solutions is  Rogue waves. These types of solutions are strong wavelets that can appear in the ocean from nowhere and disappear without a trace. In fact, it is this particular property that makes them candidates for the study of predicting catastrophic phenomena such as tsunamis, thunderstorms, and earthquakes \cite{Akhmediev}.

In order to describe more realistically  some phenomena dynamics, the space-time variation of the coefficients must be taken into account in the majority of cases. In this sense, new nonlinear Schr$\ddot{\mbox{o}}$dinger-type coupled systems with variable coefficients have been introduced \cite{Du,Rehab,Han,Kevrekidis,Xiaoyan,Manganaro,Muslum,Arash,Qiu,Yu}. The difficulty of obtaining  non-constant coefficients in the construction of  explicit solutions  can be overcome by  making use of different techniques such as the Hirota bilinear method \cite{Chakraborty,Zhang}, similarity transformations \cite{Han,Manganaro,Yu},  Darboux transformations \cite{Chai,Du} and Bäcklund transformation \cite{Chai15}, to mention a few.

During recent year and due to many applications, there has been a growing interest in the study of solitons and Rogue wave type solutions for coupled variable coefficient NLS equations \cite{Rehab,Xiaoyan,Muslum,Arash,Qiu}. At present, the field interested in these kinds of solution and moved by the necessity of data transfer are optical fiber technology and soliton transmission. More precisely, these special solutions  can overcome some difficulties related to physical properties of the fiber and the interaction of it and the transmitted soliton. Although some important advances have been made in understanding the dynamics of soliton solutions for generalized coupled NLS systems from the mathematical and physical point of view, it remains a difficult problem. In this context and under the revolution of hardware equipment and software technology, recent studies of these fascinating solutions have been carried out by using a deep learning approach as well \cite{Juncai,Raissi}.

In this work, the objective will be to study the Schr$\ddot{\mbox{o}}$dinger-type coupled system with quadratic Hamiltonian (VCNLS hereafter):
{\small 
\begin{eqnarray}
\label{VCNL1}
i\psi _{t} &=&-a(t) \psi_{xx} +\left(b(t)x^2-id(t)-xf(t)\right)\psi +i\left(g(t)-c(t)x\right)\psi_{x} +h(t)(\left\vert \varphi \right\vert ^{2s}+\left\vert \psi \right\vert
^{2s})\psi,  \\
i\varphi_{t} &=&-a(t) \varphi_{xx} +\left(b(t)x^2-id(t)-xf(t)\right)\varphi +i\left(g(t)-c(t)x\right)\varphi_{x} +h(t)(\left\vert \varphi \right\vert ^{2s}+\left\vert \psi \right\vert
^{2s})\varphi.  \label{VCNL2}
\end{eqnarray}
}Here $\psi(x,t)$ and $\varphi(x,t)$  represent  complex wavefunctions, $a(t), b(t), c(t), d(t), f(t), g(t), h(t)$ are given time-dependent functions and $s>0.$ The coefficient $a(t)$ is the dispersion management parameter, $b(t)$ determines the time-dependence of the quadratic potential, $d(t)$ is the gain/loss term, $f(t)x$ standardizes an arbitrary  linear potential and $h(t)$ is the nonlinear management. The system (\ref{VCNL1})-(\ref{VCNL2}) appears to describe the propagation of solitons in birefringence fibers, see for example \cite{Han,Yu}, and also, in the framework of solitary waves in multi-component Bose-Einstein condensates \cite{Ho,Timmermans}. We point out that the proposed system  (\ref{VCNL1})-(\ref{VCNL2}) extends many of the models studied in the literature. 

The uncoupled case ($\varphi \equiv 0$), has been the source of many studies in the last years \cite{Acosta-Humanez,Suazo16,CorderoSoto2008,Escorcia,Suazo18,Suazo}. By 
employing the similarity transformation and imposing the coefficients  satisfy a  Riccati system (see Equations (\ref{rica1})-(\ref{rica6}) and (\ref{Rica1})-(\ref{Rica6})),  soliton-like solutions (bright, dark, and Peregrine), singular solutions, and the fundamental solution ($h(t)\equiv 0$) were  constructed.
Also, the Riccati system allowed  the construction and study of explicit solutions for diffusion-type equations \cite{Suazo13} and reaction–diffusion equations \cite{Pereira}.

The paper's aim is to establish a relationship between dispersion, damping, and nonlinearity terms by the Riccati system in order to obtain  solutions such as Dark-Bright and Rogue wave soliton solutions for the coupled VCNLS (\ref{VCNL1})-(\ref{VCNL2}). The multiparameters solution of the Riccati system allows controlling the soliton dynamics; this is crucial in practice, because one needs to control the soliton transmission of information since standard vector solitons for  NLS systems don't have controllable parameters. This relation in the coefficients provides the construction of solutions with blowup in finite time  in the $L^{\infty}$ norm. Finally, an $n-$dimensional generalization of the VCNLS is studied  and the existence of soliton solutions is shown as well.

The paper is organized as follows: In Section 2 we introduce the classical coupled  Manakov system and its Rogue wave and soliton type solutions. The explicit construction of Rogue wave and soliton type-solutions for the generalized coupled NLS system (\ref{VCNL1})-(\ref{VCNL2}) is presented in Section 3. This objective is obtained by using  Theorem \ref{Th1} and solutions of the Manakov system given in the previous section. In Section 4, we construct explicit $L^{\infty}$ singular solutions for the  generalized coupled NLS system, as can be seen in Theorem \ref{Theo2}. Section 5 is devoted to the extension of the transformation established in Theorem \ref{Th1}.  Such a new transformation allows us to connect an n-dimensional coupled NLS system with a n-dimensional Manakov system, see Theorem 3. Also, soliton and Rogue wave solutions are constructed for this model. A Mathematica file has been prepared as supplementary material, verifying the Riccati systems used in the construction of the solutions in all the previous sections. Conclusions and some final remarks are given in Section 6. Section 7 corresponds to the appendix, in which the solution of the Riccati system is presented.

\section{Classical solutions of coupled NLS}

The coupled NLS system, or Manakov system 
\begin{eqnarray}
\label{ManakovSystem1}
i\chi _{\tau }+\chi _{\xi \xi }+2(\left\vert \phi \right\vert
^{2}+\left\vert \chi \right\vert ^{2})\chi &=&0, \\
\label{ManakovSystem2}
i\phi _{\tau }+\phi _{\xi \xi }+2(\left\vert \phi \right\vert
^{2}+\left\vert \chi \right\vert ^{2})\phi &=&0,
\end{eqnarray}%
is a nonlinear wave system that governs a wide variety of physical phenomena ranging from the dynamics of two-component Bose-Einstein condensate \cite{Busch2001,Kevrekidis08}, nonlinear optics \cite{Zhigang1997,Ostrovskaya1999}, and the dynamics of deep-sea waves under  the interaction of two-layer stratification  \cite{Yilmaz}.  This system (\ref{ManakovSystem1})-(\ref{ManakovSystem2}) was shown by Manakov to be integrable \cite{Manakov1974,ZAKHAROV1982}. It admits the  soliton-like structures such as dark-bright (DB) solutions. In fact, (\ref{ManakovSystem1})-(\ref{ManakovSystem2})  possesses the DB solitons, see for example \cite{Busch2001,Zhigang1997,Yu}:
\begin{equation}
\chi (\xi ,\tau )= C\tanh\left[C(\xi-2e\tau) + E \right]\exp\left[i(e\xi + (2C^2 -e^2)\tau)\right]   , \label{Dark_Brighta}
\end{equation}

\begin{equation}
\phi(\xi ,\tau )= -4C^2(a_3 + ib_3)\sech \left[C(\xi -2e \tau) + E \right]\exp\left[i\left(e \xi + (3C^2 -e^2)\tau\right) \right],\label{Dark_Brightb}
\end{equation}

\noindent where $E = \frac{1}{2}\ln\left( \frac{a_{3}^{2} + b_{3}^{2}}{2C^2}\right)$ and $C,e,a_3,b_3$ are arbitrary constants. These kinds  of solutions were first experimentally observed  in photorefractive crystals \cite{Zhigang1997}. 
Other interesting solutions with fascinating properties are the Rogue waves. Indeed, the coupled NLS system (\ref{ManakovSystem1})-(\ref{ManakovSystem2})  admits the exact  Rogue wave solution of type I \cite{BoLi,Yu}:\bigskip 
\begin{equation}
\chi (\xi ,\tau )=A \exp (i\theta _{1})\left( -1-i\sqrt{3}+\frac{%
-6AB \sqrt{3}-36\tau A^{2}\sqrt{3}-3+i(36A^{2}\tau
+6AB  +5\sqrt{3})}{12A^{2}B^{2}+8AB \sqrt{3}%
+144\tau ^{2}A^{4}+5}\right) , \label{RWIa}
\end{equation}

\begin{equation}
\phi(\xi ,\tau )=A \exp (i\theta_{2})\left( -1+i\sqrt{3}+\frac{%
-6A B \sqrt{3}+36\tau A^{2}\sqrt{3}-3+i(36A^{2}\tau
-6AB  -5\sqrt{3})}{12A^{2}B^{2}+8AB \sqrt{3}%
+144\tau ^{2}A^{4}+5}\right) ,\label{RWIb}
\end{equation}
where $B =\xi +6q\tau ,$ $\theta _{i}=d_{i}\xi
+(2c_{1}^{2}+2c_{2}^{2}-d_{i}^{2})\tau $ with $i = 1,2.$ The parameters $A =d_{2}+3q,$ $d_{1}=d_{2}-2A ,$ $c_{i}=\pm 2A ,$ with $i = 1,2$ and $d_{2},q$ are
arbitrary constants. Similarly, the  Rogue wave solution of type II \cite{BoLi,Yu} is given by:

\begin{equation}
\chi (\xi ,\tau )=A \left( -1-i\sqrt{3}+\frac{G_{1}+iH_{1}}{D}\right) \exp
(i\theta _{1}), \label{RWIIa}
\end{equation}

\begin{equation}
\phi(\xi ,\tau )=A\left( -1+i\sqrt{3}+\frac{G_{2}+iH_{2}}{D}\right) \exp
(i\theta _{2}),\label{RWIIb}
\end{equation}

\noindent where {\tiny 
\begin{eqnarray*}
D &=&1+4\sqrt{3}AB +24A^{2}B^{2}+16\sqrt{3}A^{3}B^{3}+12A^{4}B^{4}+48A^{4}(9+8\sqrt{3}A
B +6A^{2}B^{2})\tau ^{2}+1728A^{8}\tau ^{4}, \\
G_{1} &=&-3\left( -1+6A^{2}B^{2}+4\sqrt{3}A^{3}B
^{3}+4A ^{2}(\sqrt{3}+12A B +6\sqrt{3}A ^{2}B
^{2})\tau +24A ^{4}(3+2\sqrt{3}A B )\tau ^{2}+288\sqrt{3}%
A ^{6}\tau ^{3}\right) , \\
G_{2} &=&3\left( 1-6A ^{2}B^{2}-4\sqrt{3}A ^{3}B
^{3}+4A ^{2}(\sqrt{3}+12A B +6\sqrt{3}A ^{2}B
^{2})\tau -24A ^{4}(3+2\sqrt{3}A B )\tau ^{2}+288\sqrt{3}%
A ^{6}\tau ^{3}\right) , \\
H_{1} &=&\sqrt{3}+12A B +18\sqrt{3}A ^{2}B ^{2}+12A
^{3}B ^{3}+12A ^{2}(9+8\sqrt{3}A B +6A ^{2}B
^{2})\tau +24A ^{4}(13\sqrt{3}+6A B )\tau ^{2}+864A
^{6}\tau ^{3}, \\
H_{2} &=&-\sqrt{3}-12A B -18\sqrt{3}A ^{2}B
^{2}-12A ^{3}B ^{3}+12A ^{2}(9+8\sqrt{3}A B
+6A ^{2}B ^{2})\tau -24A ^{4}(13\sqrt{3}+6A B )\tau
^{2}+864A ^{6}\tau ^{3}.
\end{eqnarray*}
}

In the subsequent  section, we will use these solutions to construct explicit solutions for the generalized coupled NLS  system.


\section{Generalized Coupled NLS System}

The aim of this section is to establish the existence of soliton-type and Rogue wave solutions (by the explicit construction of these solutions) of the generalized coupled NLS system (\ref{VCNL1})-(\ref{VCNL2}). More precisely, by using similarity transformations and by imposing the coefficients to satisfy the Riccati system, classical solutions of the standard Manakov system (\ref{ManakovSystem1})-(\ref{ManakovSystem2}) are extended to this general framework. Indeed, the main result of this section is:

\begin{theorem}\label{Th1}
Let $l_{0}=\pm 1.$ Then,  the following NLS coupled system (or VCNLS)
{\small 
\begin{eqnarray}
i\psi _{t} &=&-a(t)\psi_{xx} +(b\left(
t\right) x^{2}-id(t)-xf(t))\psi +i(g(t)-c\left( t\right) x)\psi
_{x}+h(t)(\left\vert \varphi \right\vert ^{2}+\left\vert \psi \right\vert
^{2})\psi,  \label{vcNLS1} \\
i\varphi_{t} &=&-a(t)\varphi_{xx} +(b\left(
t\right) x^{2}-id(t)-xf(t))\varphi +i(g(t)-c\left( t\right) x)\varphi
_{x}+h(t)(\left\vert \varphi \right\vert ^{2}+\left\vert \psi \right\vert
^{2})\varphi,  \label{vcNLS2}
\end{eqnarray}%
} admits dark-bright (DB) solitons and Rogue wave type
solutions, where $h(t)$ satisfies the integrability condition (\ref{integrability}).
\end{theorem}
\begin{proof}

We look for solutions for (\ref{vcNLS1})-(\ref{vcNLS2}) of the form 
\begin{equation}
\psi (x,t)=\frac{1}{\sqrt{\mu (t)}}e^{i(\alpha (t)x^{2}+\delta (t)x+\kappa
(t))}\chi (\xi ,\tau ),\qquad \xi =\beta (t)x+\varepsilon (t),\qquad \tau
=\gamma (t),  \label{form1}
\end{equation}

and 
\begin{equation}
\varphi (x,t)=\frac{1}{\sqrt{\mu (t)}}e^{i(\alpha (t)x^{2}+\delta
(t)x+\kappa (t))}\phi (\xi ,\tau ).\qquad  \label{form2}
\end{equation}

\noindent After substituting (\ref{form1})-(\ref{form2}) in (\ref{vcNLS1})-(\ref%
{vcNLS2}) and taking the real and imaginary parts, we obtain the Riccati system   \cite{CorderoSoto2008,Escorcia,Suazo18,Suazo,Suazo13}: 
\begin{equation}
\dfrac{d\alpha }{dt}+b(t)+2c(t)\alpha +4a(t)\alpha^{2}=0,  \label{rica1}
\end{equation}%
\begin{equation}
\dfrac{d\beta }{dt}+(c(t)+4a(t)\alpha(t))\beta =0,  \label{rica2}
\end{equation}%
\begin{equation}
\dfrac{d\gamma }{dt}+a(t)\beta^{2}(t)l_{0}=0,
\label{rica3}
\end{equation}%
\begin{equation}
\dfrac{d\delta }{dt}+(c(t)+4a(t)\alpha(t))\delta = f(t)+2\alpha (t)g(t),
\label{rica4}
\end{equation}%
\begin{equation}
\dfrac{d\varepsilon }{dt}=(g(t)-2a(t)\delta(t))\beta (t),  \label{rica5}
\end{equation}%
\begin{equation}
\dfrac{d\kappa }{dt}=g(t)\delta (t)-a(t)\delta^{2}(t).  \label{rica6}
\end{equation}%
\ Considering the standard substitution\ 
\begin{equation}
\alpha = \left(\dfrac{1}{4a(t)}\dfrac{\mu ^{\prime }(t)}{\mu (t)}-\dfrac{d(t)}{2a(t)%
}\right),  \label{sus1}
\end{equation}%
it follows that the Riccati equation (\ref{rica1}) becomes\ 
\begin{equation}
\mu ^{\prime \prime }-\eta (t)\mu ^{\prime }+4\sigma (t)\mu =0,
\label{Carac1}
\end{equation}%
with\ 
\begin{equation}
\eta (t)=\frac{a^{\prime }}{a}-2c+4d,\hspace{1cm}\sigma (t)=ab-cd+d^{2}+%
\frac{d}{2}\left( \frac{a^{\prime }}{a}-\frac{d^{\prime }}{d}\right) .
\end{equation}%
\noindent  We will refer to (\ref{Carac1}) as the characteristic equation of the
Riccati system. Further, if we
choose by integrability conditions for the uncoupled case 
\begin{equation}
h(t)=\lambda a(t)\beta^{2}(t)\mu(t)\label{integrability}    
\end{equation}
 with $\lambda \in \mathbb{R}$, the functions $\chi $ and $\phi $ will satisfy the NLS system in the new  variables $\tau$ and $\xi$:
\begin{eqnarray}
i\chi _{\tau }-l_{0}\chi_{\xi \xi }+l_{0}\lambda(\left\vert \phi \right\vert
^{2}+\left\vert \chi \right\vert ^{2})\chi  &=&0,  \label{Classical 1} \\
i\phi _{\tau }-l_{0}\phi _{\xi \xi }+l_{0}\lambda(\left\vert \phi \right\vert
^{2}+\left\vert \chi \right\vert ^{2})\phi  &=&0,  \label{Cassical 2}
\end{eqnarray}
where  $\alpha (t),$ $\beta (t),$ $\gamma (t),$ $\delta (t),$ $\epsilon (t)$
and $\kappa (t)$ are given by the Equations (\ref{mu})-(\ref{kappa0}). Therefore, the existence of these kinds of solutions (DB, and Rogue wave solutions) are established by letting $\chi$ and $\phi$ be the solutions discussed in the previous section, see Equations (\ref{Dark_Brighta})-(\ref{RWIIb}). 
\end{proof}


\subsection{Rogue wave solutions}

The existence of Rogue wave solutions can be guaranteed by the extension of the solutions (\ref{RWIa})-(\ref{RWIb}) and (\ref{RWIIa})-(\ref{RWIIb})  by the previous Theorem \ref{Th1}. We start with the Type I Rogue Wave.

\subsubsection{Case $a = d = 1, \ b = t,  \ c = f = g = 0, \ h = \frac{-2e^{2t}}{{}_{0}F_{1}\left(2/3,-4t^{3}/9 \right)}  $}

In this case, the coupled nonlinear Schr$\ddot{\mbox{o}}$dinger equations have the form
\begin{equation}
i\psi_{t}=-\psi_{xx}+tx^{2}\psi-i\psi-\frac{2e^{2t}}{{}_{0}F_{1}\left(2/3,-4t^{3}/9 \right)}\left(|\varphi|^{2}+|\psi|^{2}\right)\psi, \label{Ex1a}
\end{equation} 
\begin{equation}
i\varphi_{t}=-\varphi_{xx}+tx^{2}\varphi-i\varphi-\frac{2e^{2t}}{{}_{0}F_{1}\left(2/3,-4t^{3}/9 \right)}\left(|\varphi|^{2}+|\psi|^{2}\right)\varphi. \label{Ex1b}
\end{equation} 
Then, according to Theorem \ref{Th1}, we get 
\begin{eqnarray*}
 \alpha(t)=\frac{1}{4t}-\frac{1+t^{3}{}_{0}F_{1}\left(2/3,-4t^{3}/9 \right){}_{0}F_{1}\left(7/3,-4t^{3}/9\right)}{4t{}_{0}F_{1}\left(2/3,-4t^{3}/9 \right){}_{0}F_{1}\left(4/3,-4t^{3}/9 \right)},  \quad \quad \quad \beta(t)= \dfrac{1}{{}_{0}F_{1}\left(2/3,-4t^{3}/9 \right)}, \\\\
 \gamma(t)=\frac{t{}_{0}F_{1}\left(4/3,-4t^{3}/9 \right)}{{}_{0}F_{1}\left(2/3,-4t^{3}/9 \right)},  \quad \quad \quad \delta(t)=\dfrac{1}{{}_{0}F_{1}\left(2/3,-4t^{3}/9 \right)}, \quad \quad \varepsilon(t)=-\dfrac{2t{}_{0}F_{1}\left(4/3,-4t^{3}/9 \right)}{{}_{0}F_{1}\left(2/3,-4t^{3}/9 \right)}, \\\\
 \kappa(t)=-\frac{t{}_{0}F_{1}\left(4/3,-4t^{3}/9 \right)}{{}_{0}F_{1}\left(2/3,-4t^{3}/9 \right)}, \quad \quad \mu(t)=e^{2t}{}_{0}F_{1}\left(2/3,-4t^{3}/9 \right), \quad \quad \quad \quad 
\end{eqnarray*}
where ${}_{0}F_{1}\left(a,z\right)$ is the confluent hypergeometric function given by the series ${}_{0}F_{1}\left(a,z\right)=\sum_{k=0}^{\infty}\frac{z^{k}}{(a)_{k}k!}$ with $(a)_{k}=\Gamma(a+k)/\Gamma(a)$ and $\Gamma(z)$ is the gamma function. Here, we have solved the Riccati system 

\begin{figure}[h!]
\centering
\subfigure[Profile of $|\psi|^2$.]{\includegraphics[scale=0.34]{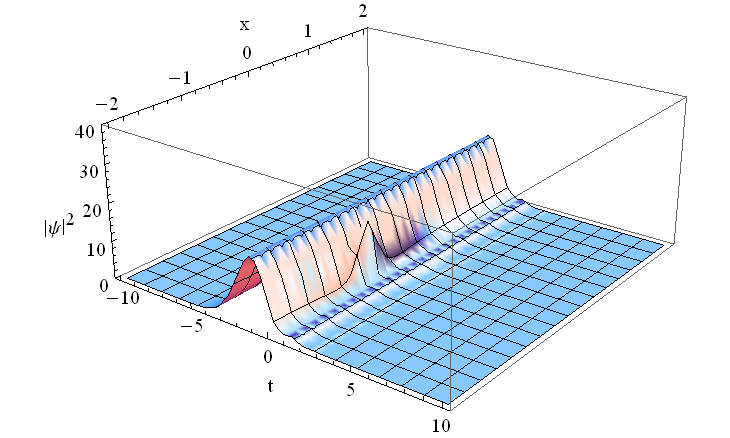}}
\subfigure[Profile of $|\varphi|^2$.]{\includegraphics[scale=0.32]{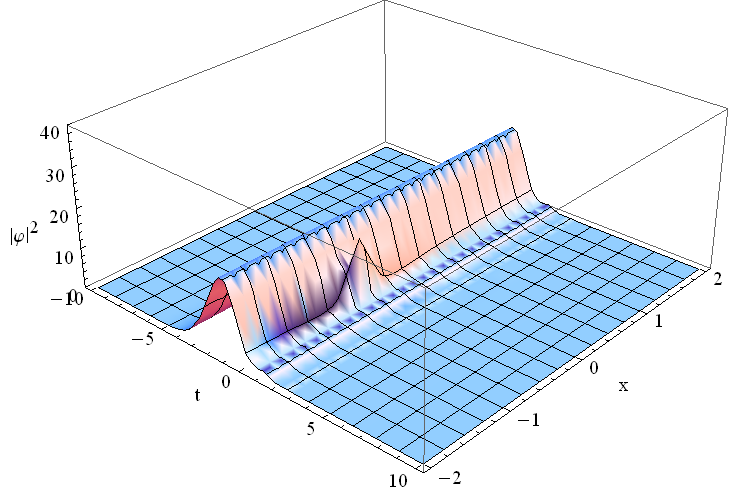}}
\subfigure[Difference $|\psi-\varphi|$.]{\includegraphics[scale=0.26]{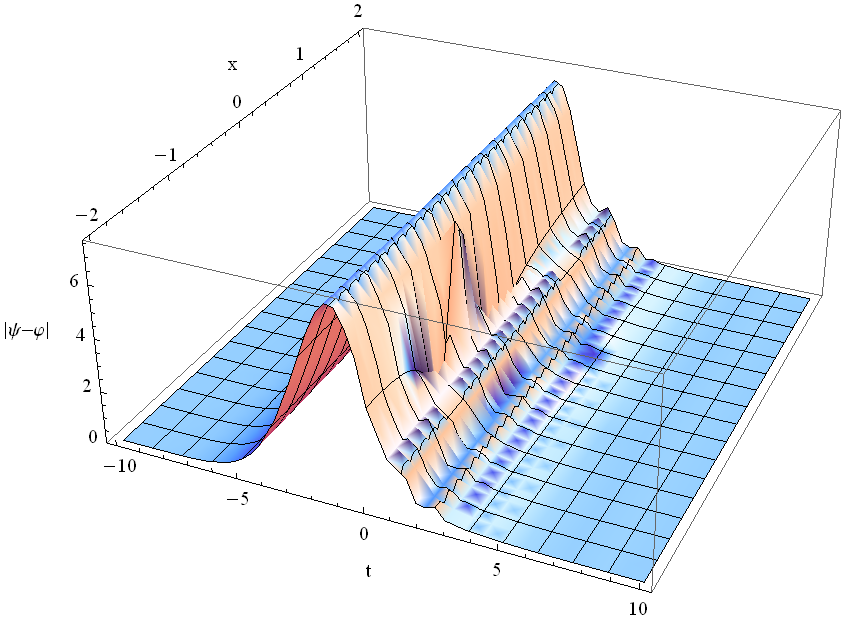}}
\caption{Solutions for the system  (\ref{Ex1a})-(\ref{Ex1b}) for the parameters $d_2=1,$ $q=0$.}\label{Fig1}
\end{figure}
\noindent with the initial conditions $\alpha(0)=0$, $\beta(0)=1$, $\gamma(0)=0,$ $\delta(0)=1$, $\varepsilon(0)=0,$ $\kappa(0)=0$ and $\mu(0)=1.$ \\\
Therefore, we  can construct a solution in the form:
\begin{equation}
\psi(x,t)=\dfrac{e^{-t}}{\sqrt{{}_{0}F_{1}\left(2/3,-4t^{3}/9 \right)}}\exp\left[i\left(\alpha(t)x^2+\delta(t)x+\kappa(t)\right)\right]\chi(\xi,\tau),
\end{equation}
\begin{equation}
\varphi(x,t)=\dfrac{e^{-t}}{\sqrt{{}_{0}F_{1}\left(2/3,-4t^{3}/9 \right)}}\exp\left[i\left(\alpha(t)x^2+\delta(t)x+\kappa(t)\right)\right]\phi(\xi,\tau),
\end{equation}\\\
where $\chi$ and $\phi$ are given by Equations (\ref{RWIa})-(\ref{RWIb}). The profiles of these solutions for the parameters $d_2 = 1$ and $q = 0$ are shown in Figure \ref{Fig1}(a)-(b). Figure \ref{Fig1}(c) shows the difference between the two solutions to show that they are different.


\subsubsection{Case $a = -\frac{\cos t}{2}, \ b = c = f = g = 0, \ d = -1, \ h = e^{-2t}\cos t $}
The corresponding  system is given by
\begin{equation}
i\psi_{t}=\frac{\cos t}{2}\psi_{xx}+i\psi+e^{-2t}\cos t \left(|\varphi|^{2}+|\psi|^{2}\right)\psi, \label{Ex2a}
\end{equation} 
\begin{equation}
i\varphi_{t}=\frac{\cos t}{2}\varphi_{xx}+i\varphi+e^{-2t}\cos t \left(|\varphi|^{2}+|\psi|^{2}\right)\varphi. \label{Ex2b}
\end{equation} 

\begin{figure}[h!]
\centering
\subfigure[Profile of $|\psi|^2$.]{\includegraphics[scale=0.29]{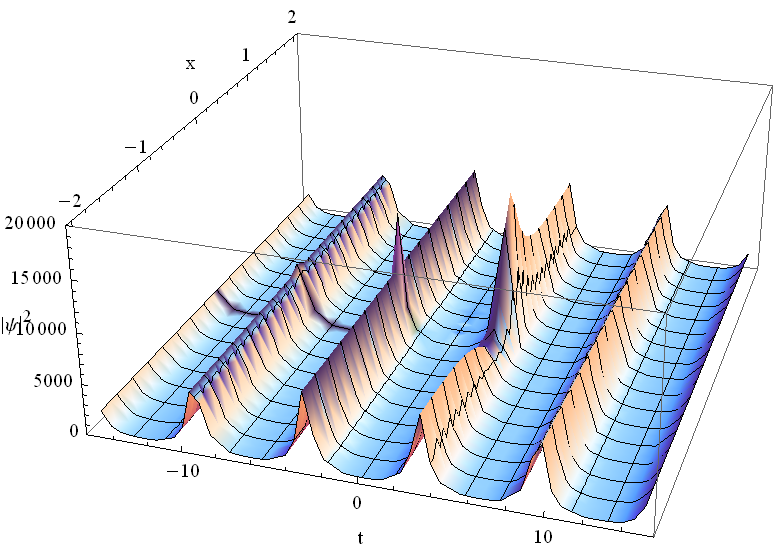}}
\subfigure[Profile of $|\varphi|^2$.]{\includegraphics[scale=0.28]{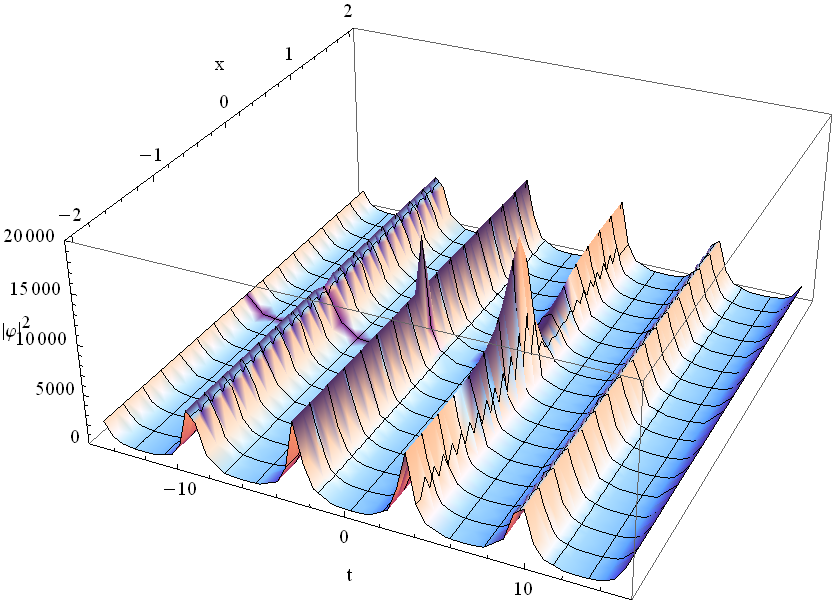}}
\subfigure[Difference $|\psi-\varphi|$.]{\includegraphics[scale=0.29]{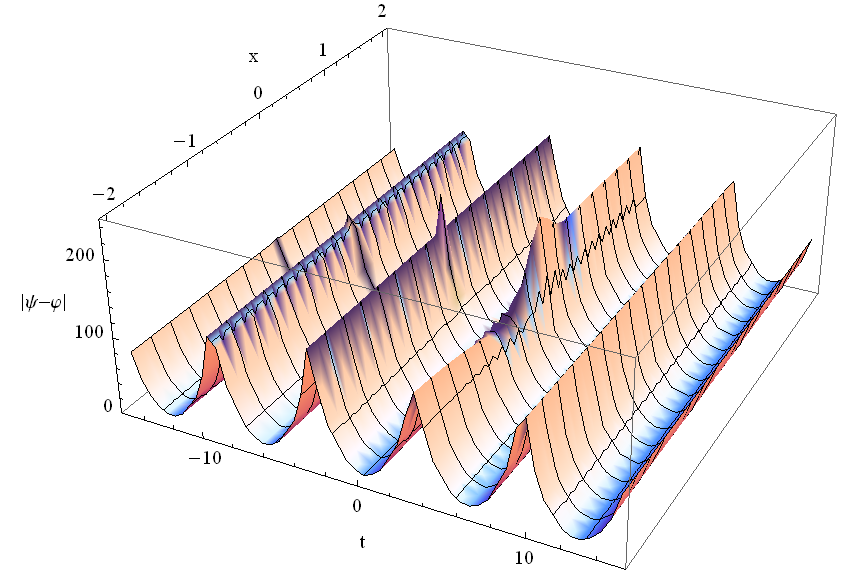}}
\caption{Solutions for the system  (\ref{Ex2a})-(\ref{Ex2b}) with  parameters $d_2=1,$ $q=-1$.}\label{Fig2}
\end{figure}

 Then, 
\begin{eqnarray*}
 \alpha(t)= 0,  \quad \quad \quad \quad \beta(t)= 1, \quad \quad \quad 
 \gamma(t)=-\frac{\sin t}{2},  \quad \quad \quad   \delta(t)=1, \\
 \varepsilon(t)= \sin t, \quad \quad \quad \quad \quad \quad 
 \kappa(t)=\frac{\sin t}{2},\quad \quad \quad \quad \quad  \quad \quad \mu(t)=e^{-2t}.\quad \quad \quad \quad 
\end{eqnarray*}
 The initial conditions used to solve the Riccati system are the same as those given in the previous case. In fact, except for  the section on bending dynamics, we shall use these initial conditions for the rest of the study. Solutions for the system (\ref{Ex2a})-(\ref{Ex2b}) are

\begin{equation}
\psi(x,t)=e^{\arccos(\cos t)}\exp\left[i\left(\alpha(t)x^2+\delta(t)x+\kappa(t)\right)\right]\chi(\xi,\tau),
\end{equation}
\begin{equation}
\varphi(x,t)=e^{\arccos(\cos t)}\exp\left[i\left(\alpha(t)x^2+\delta(t)x+\kappa(t)\right)\right]\phi(\xi,\tau).
\end{equation}\\\
Again $\chi$ and $\phi$ are given by the Equation (\ref{RWIa})-(\ref{RWIb}). Figure \ref{Fig2} shows the time evolution of $|\psi|^2$ and $|\varphi|^2$ for the corresponding solutions. The difference between the two solutions are shown in the Figure \ref{Fig2}(c). 


 The next two cases correspond to the Type II Rogue Wave solutions.

 \subsubsection{Case $a = -\sech t/2, \ b = -4a,  \ c = 4a, \ d = 2a, \ f = g = 0,  \ h = \frac{-2\sech t}{-2\cosh\left(\sqrt{8}\gd t\right) + \sqrt{2}\sinh\left(\sqrt{8}\gd t\right)}  $}
  Assume the system
\begin{equation}
i\psi_{t}=\sech t\left(\dfrac{\psi_{xx}}{2}+(i+2x^2)\psi +2ix\psi_{x}-\frac{2\left(|\varphi|^{2}+|\psi|^{2}\right)\psi}{-2\cosh\left(\sqrt{8}\gd t\right) + \sqrt{2}\sinh\left(\sqrt{8}\gd t\right)}
\right), \label{Ex3a}
\end{equation} 
\begin{equation}
i\varphi_{t}=\sech t\left(\dfrac{\varphi_{xx}}{2}+(i+2x^2)\varphi +2ix\varphi_{x}-\frac{2\left(|\varphi|^{2}+|\psi|^{2}\right)\varphi}{-2\cosh\left(\sqrt{8}\gd t\right) + \sqrt{2}\sinh\left(\sqrt{8}\gd t\right)}
\right), \label{Ex3b}
\end{equation}
where $\gd t$ is the Gudermannian function given by $\gd t=2\arctan\left[\tanh\left(\tfrac12t\right)\right]$. The expressions for the functions $\alpha(t)$, $\beta(t)$, $\gamma(t)$, $\delta(t)$, $\varepsilon(t)$, $\kappa(t)$ and $\mu(t)$ are
\begin{eqnarray*}
 \alpha(t)=\dfrac{1}{1-\sqrt{2}\coth\left(\sqrt{8}\gd t\right)},  \quad \quad \quad \beta(t)= \dfrac{\sqrt{2}}{\sqrt{2}\cosh\left(\sqrt{8}\gd t\right)-\sinh\left(\sqrt{8}\gd t\right)}, \\\\
 \gamma(t)=\dfrac{1}{4-4\sqrt{2}\coth\left(\sqrt{8}\gd t\right)},  \quad \quad \quad  \quad  \quad \delta(t)=\dfrac{\sqrt{2}}{\sqrt{2}\cosh\left(\sqrt{8}\gd t\right)-\sinh\left(\sqrt{8}\gd t\right)},  
 \\\ \varepsilon(t)=-\dfrac{1}{2-\sqrt{8}\coth\left(\sqrt{8}\gd t\right)}, \quad \quad \quad \quad \kappa(t)=-\dfrac{1}{4-4\sqrt{2}\coth\left(\sqrt{8}\gd t\right)},
  \\\\ \mu(t)=\cosh\left(\sqrt{8}\gd t\right)-\frac{\sqrt{2}}{2}\sinh\left(\sqrt{8}\gd t\right). \quad \quad \quad \quad  \quad \quad 
\end{eqnarray*}

\begin{figure}[h!]
\centering
\subfigure[Profile of $|\psi|^2$.]{\includegraphics[scale=0.32]{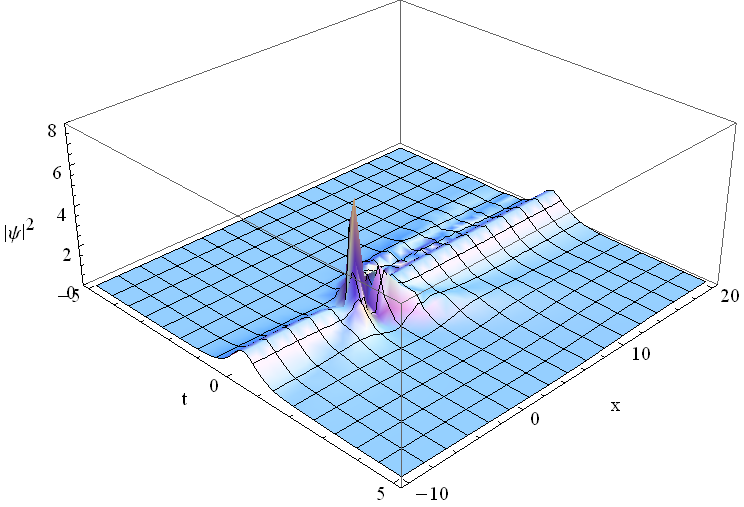}}
\subfigure[Profile of $|\varphi|^2$.]{\includegraphics[scale=0.31]{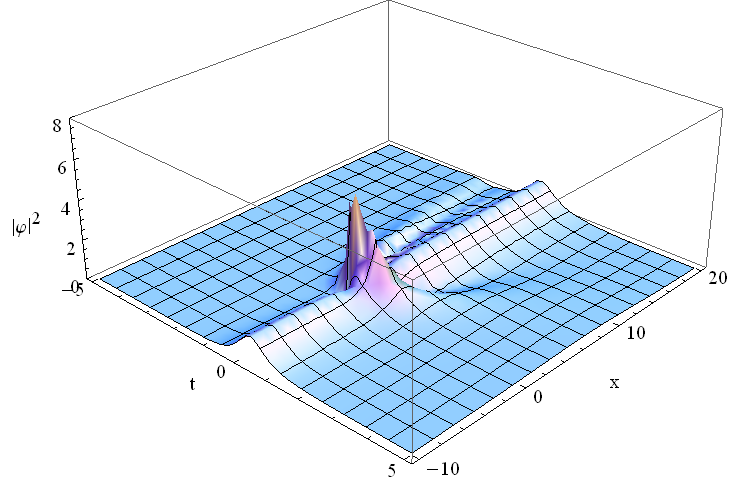}}
\subfigure[Difference $|\psi-\varphi|$.]{\includegraphics[scale=0.31]{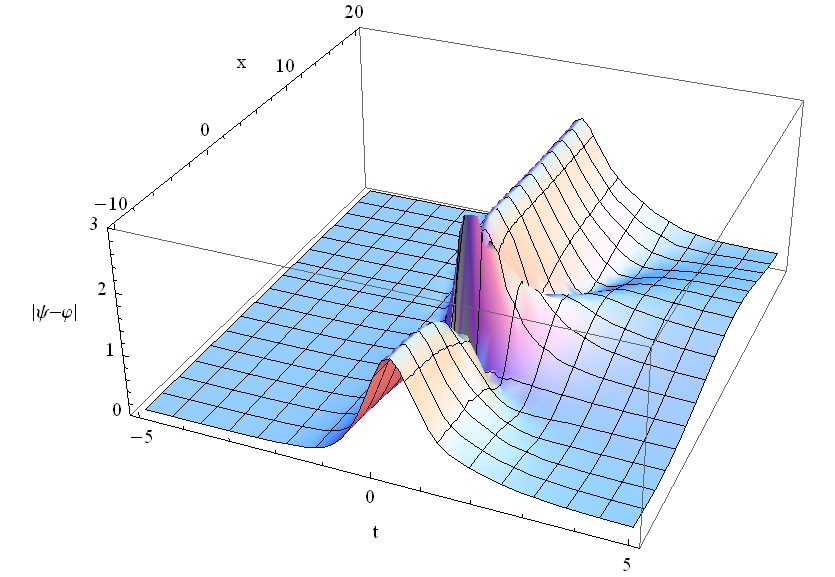}}
\caption{Solutions for the system  (\ref{Ex3a})-(\ref{Ex3b}) for the parameters $d_2=-0.5,$ $q=0$.}\label{Fig3}
\end{figure}
\noindent Now, we can give a form for the solutions, that is,
\begin{equation}
\psi(x,t)=\dfrac{1}{\sqrt{\cosh\left(\sqrt{8}\gd t\right)-\frac{\sqrt{2}}{2}\sinh\left(\sqrt{8}\gd t\right)}}\exp\left[i\left(\alpha(t)x^2+\delta(t)x+\kappa(t)\right)\right]\chi(\xi,\tau),
\end{equation}
\begin{equation}
\varphi(x,t)=\dfrac{1}{\sqrt{\cosh\left(\sqrt{8}\gd t\right)-\frac{\sqrt{2}}{2}\sinh\left(\sqrt{8}\gd t\right)}}\exp\left[i\left(\alpha(t)x^2+\delta(t)x+\kappa(t)\right)\right]\phi(\xi,\tau),
\end{equation}
with  $\chi$ and $\phi$ satisfying the Equations (\ref{RWIIa})-(\ref{RWIIb}). The dynamics of these solutions are shown in Figure \ref{Fig3}.


\subsubsection{Case $a = -(1+t^2) \cos t /2, \ b = 2a,  \ c = 4a, \ d = 2a, \ f = g = 0,  \ h = \frac{2a}{\cosh \theta -\sqrt{2}\sinh \theta }   $}
Consider the system of the nonlinear Schr$\ddot{\mbox{o}}$dinger equations
\begin{equation}
i\psi_{t}=(1+t^2)\cos t\left(\dfrac{\psi_{xx}}{2}+(i-x^2)\psi +2ix\psi_{x}-\frac{\left(|\varphi|^{2}+|\psi|^{2}\right)\psi}{\cosh \theta -\sqrt{2}\sinh \theta }\right), \label{Ex4a}
\end{equation} 
\begin{equation}
i\varphi_{t}=(1+t^2)\cos t\left(\dfrac{\varphi_{xx}}{2}+(i-x^2)\varphi +2ix\varphi_{x}-\frac{\left(|\varphi|^{2}+|\psi|^{2}\right)\varphi}{\cosh \theta -\sqrt{2}\sinh \theta}\right), \label{Ex4b}
\end{equation}
with $\theta = \sqrt{8}t \cos t + \sqrt{2}(t^2 -1)\sin t$. In this case, the solution of the Riccati system has the form
\begin{eqnarray*}
 \alpha(t)=\dfrac{1}{-2+\sqrt{2}\coth \theta}, \quad \quad  \quad \quad \quad \beta(t)=\dfrac{1}{\cosh \theta - \sqrt{2}\sinh \theta} , \\\\
 \gamma(t)=\dfrac{1}{4-\sqrt{8}\coth \theta},  \quad \quad \quad  \quad  \quad \delta(t)=\dfrac{1}{\cosh \theta -\sqrt{2}\sinh \theta},  
 \\\ \varepsilon(t)=\dfrac{1}{-2 + \sqrt{2}\coth \theta}, \quad \quad \quad \quad \kappa(t)=\dfrac{1}{-4 + \sqrt{8}\coth \theta},
  \\\\ \mu(t)= \cosh \theta - \sqrt{2}\sinh \theta .  \quad \quad \quad \quad  \quad \quad \quad \quad 
\end{eqnarray*}

\begin{figure}[h!]
\centering
\subfigure[Dynamics of $|\psi|^2$.]{\includegraphics[scale=0.31]{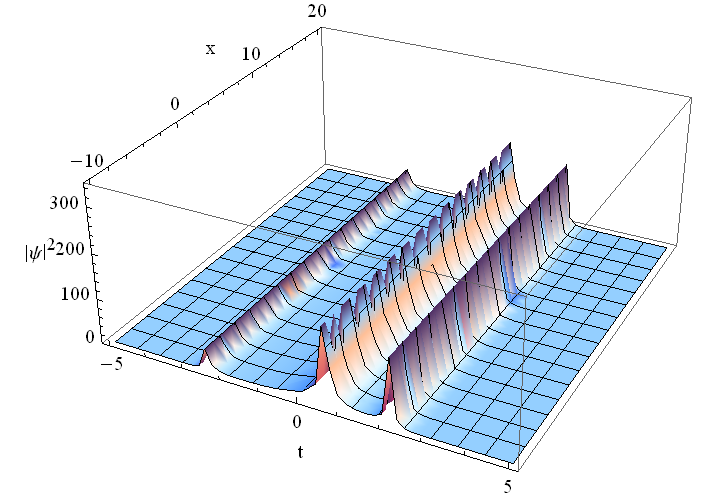}}
\subfigure[Dynamics of $|\varphi|^2$.]{\includegraphics[scale=0.29]{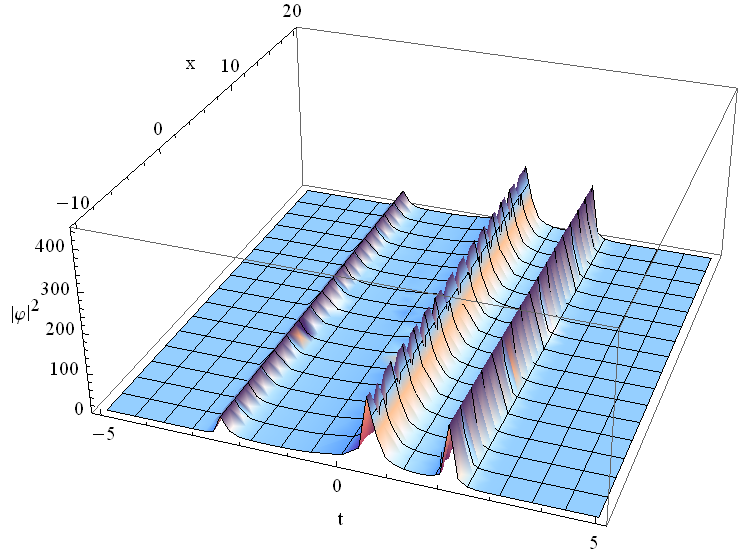}}
\subfigure[Difference $|\psi-\varphi|$.]{\includegraphics[scale=0.28]{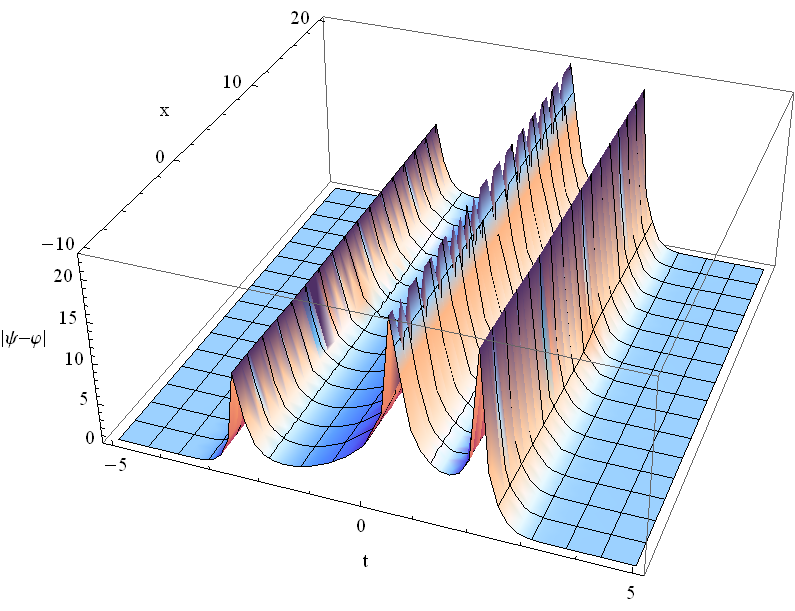}}
\caption{Solutions for the system  (\ref{Ex4a})-(\ref{Ex4b}) with parameters $d_2=1,$ $q=-1$.}\label{Fig4}
\end{figure}
 \noindent Then, by the assumption that $\chi$ and $\phi$ satisfy equations (\ref{RWIIa})-(\ref{RWIIb}), Theorem \ref{Th1} provides the following solutions:

\begin{equation}
\psi(x,t)=\dfrac{1}{\sqrt{\cosh \theta - \sqrt{2}\sinh \theta}}\exp\left[i\left(\alpha(t)x^2+\delta(t)x+\kappa(t)\right)\right]\chi(\xi,\tau),
\end{equation}
\begin{equation}
\varphi(x,t)=\dfrac{1}{\sqrt{\cosh \theta - \sqrt{2}\sinh \theta}}\exp\left[i\left(\alpha(t)x^2+\delta(t)x+\kappa(t)\right)\right]\phi(\xi,\tau).
\end{equation}

\noindent Figure \ref{Fig4} shows the profiles of   $|\psi|^2,  |\varphi|^2$  for these particular solutions.

\subsection{Dark-bright soliton-type solutions}

This part of the work is concerned with the construction of  solutions  for the generalized coupled NLS system  with soliton-type properties.  As we will see, (\ref{vcNLS1})-(\ref{vcNLS2}) admits  DB type-solitons. 

\subsubsection{Case $a = 2, \ b = \frac{t+1}{2},  \ c = -2, \ d = -1, \ f = g = 0,  \ h = \frac{-4}{{}_{0}F_{1}\left(2/3,-4t^{3}/9 \right)-2t{}_{0}F_{1}\left(4/3,-4t^{3}/9 \right)}    $}
The system of the nonlinear Schr$\ddot{\mbox{o}}$dinger equations
\begin{equation}
i\psi_{t}=-2\psi_{xx} + \frac{t+1}{2}x^{2}\psi + i\psi + 2ix\psi_{x}- \frac{4\left(|\varphi|^{2}+|\psi|^{2}\right)\psi}{{}_{0}F_{1}\left(2/3,-4t^{3}/9 \right)-2t{}_{0}F_{1}\left(4/3,-4t^{3}/9 \right)}, \label{Ex5a}
\end{equation} 
\begin{equation}
i\varphi_{t}=-2\varphi_{xx} + \frac{t+1}{2}x^{2}\varphi + i\varphi + 2ix\varphi_{x}- \frac{4\left(|\varphi|^{2}+|\psi|^{2}\right)\varphi}{{}_{0}F_{1}\left(2/3,-4t^{3}/9 \right)-2t{}_{0}F_{1}\left(4/3,-4t^{3}/9 \right)}, \label{Ex5b}
\end{equation}

\begin{figure}[h!]
\centering
\subfigure[Profile of $|\psi|^2$.]{\includegraphics[scale=0.29]{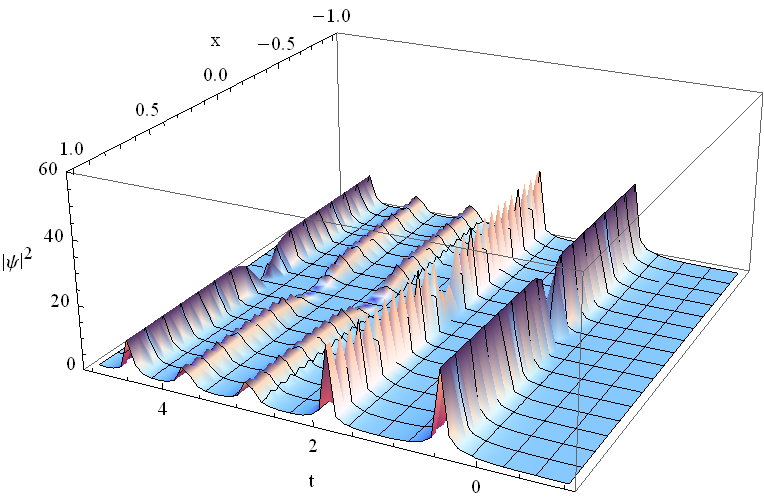}}
\subfigure[Profile of $|\varphi|^2$.]{\includegraphics[scale=0.3]{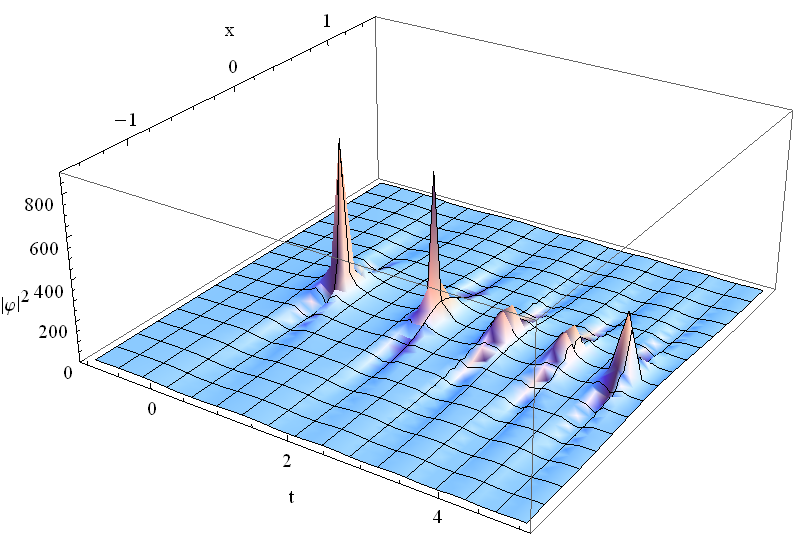}}
\caption{Solutions for the system  (\ref{Ex5a})-(\ref{Ex5b}) for the parameters $e=-1,$ $C=1,$ $a_3 = 1$ and $b_3 = 1$.}\label{Fig5}
\end{figure}

\noindent admits the functions
\begin{eqnarray*}
 \alpha(t)= \dfrac{1}{4} + \dfrac{1}{8t} + \dfrac{1}{-8t{}_{0}F_{1}\left(2/3,-4t^{3}/9 \right){}_{0}F_{1}\left(4/3,-4t^{3}/9 \right)  + 16t^{2}{}_{0}F_{1}^{2}\left(4/3,-4t^{3}/9 \right)}-\dfrac{t^{2}{}_{0}F_{1}\left(7/3,-4t^{3}/9 \right)}{8{}_{0}F_{1}\left(4/3,-4t^{3}/9 \right)}, \\\ \beta(t)=\dfrac{1}{{}_{0}F_{1}\left(2/3,-4t^{3}/9 \right)-2t{}_{0}F_{1}\left(4/3,-4t^{3}/9 \right)}, \quad \quad 
 \gamma(t)=\dfrac{2t{}_{0}F_{1}\left(4/3,-4t^{3}/9 \right)}{{}_{0}F_{1}\left(2/3,-4t^{3}/9 \right)-2t{}_{0}F_{1}\left(4/3,-4t^{3}/9 \right)},    \\\ \delta(t)=\dfrac{1}{{}_{0}F_{1}\left(2/3,-4t^{3}/9 \right)-2t{}_{0}F_{1}\left(4/3,-4t^{3}/9 \right)},  
 \quad \quad  \varepsilon(t)=\dfrac{-4t{}_{0}F_{1}\left(4/3,-4t^{3}/9 \right)}{{}_{0}F_{1}\left(2/3,-4t^{3}/9 \right)-2t{}_{0}F_{1}\left(4/3,-4t^{3}/9 \right)},
\end{eqnarray*}

\begin{eqnarray*}
\kappa(t)=\dfrac{-2t{}_{0}F_{1}\left(4/3,-4t^{3}/9 \right)}{{}_{0}F_{1}\left(2/3,-4t^{3}/9 \right)-2t{}_{0}F_{1}\left(4/3,-4t^{3}/9 \right)},
\quad \quad  \mu(t)= {}_{0}F_{1}\left(2/3,-4t^{3}/9 \right)-2t{}_{0}F_{1}\left(4/3,-4t^{3}/9 \right).
\end{eqnarray*}
\noindent In these terms,  explicit solutions are given by the expressions
\begin{equation}
\psi(x,t)=\dfrac{1}{\sqrt{{}_{0}F_{1}\left(2/3,-4t^{3}/9 \right)-2t{}_{0}F_{1}\left(4/3,-4t^{3}/9 \right)}}\exp\left[i\left(\alpha(t)x^2+\delta(t)x+\kappa(t)\right)\right]\chi(\xi,\tau),
\end{equation}
\begin{equation}
\varphi(x,t)=\dfrac{1}{\sqrt{{}_{0}F_{1}\left(2/3,-4t^{3}/9 \right)-2t{}_{0}F_{1}\left(4/3,-4t^{3}/9 \right)}}\exp\left[i\left(\alpha(t)x^2+\delta(t)x+\kappa(t)\right)\right]\phi(\xi,\tau),
\end{equation}
with $\chi$ and $\phi$ satisfying  Equations (\ref{Dark_Brighta})-(\ref{Dark_Brightb}).

\subsubsection{Case $a = 1, \ b = \frac{\sin^2 t - \cos t}{4},  \ c = -\sin t, \ d = \sin t, \ f = g = 0,  \ h = - 2e^{3-3\cos t}    $}
Consider the system of nonlinear Schr$\ddot{\mbox{o}}$dinger equations defined for $t>0,$
\begin{equation}
i\psi_{t}=-\psi_{xx} + \frac{\sin^{2} t-\cos t}{4}x^{2}\psi - i\sin t \psi + ix\sin t \psi_{x}- 2e^{3-3\cos t}\left(|\varphi|^{2}+|\psi|^{2}\right)\psi, \label{Ex6a}
\end{equation} 
\begin{equation}
i\varphi_{t}=-\varphi_{xx} + \frac{\sin^{2} t-\cos t}{4}x^{2}\varphi - i\sin t \varphi + ix\sin t \varphi_{x}- 2e^{3-3\cos t}\left(|\varphi|^{2}+|\psi|^{2}\right)\varphi. \label{Ex6b}
\end{equation}
It has the explicit soliton-type solutions
\begin{equation}
\psi(x,t)=\dfrac{1}{\sqrt{e^{3-3\cos t}}}\exp\left[i\left(\alpha(t)x^2+\delta(t)x+\kappa(t)\right)\right]\chi(\xi,\tau),
\end{equation}
\begin{equation}
\varphi(x,t)=\dfrac{1}{\sqrt{e^{3-3\cos t}}}\exp\left[i\left(\alpha(t)x^2+\delta(t)x+\kappa(t)\right)\right]\phi(\xi,\tau),
\end{equation}

\begin{figure}[h!]
\centering
\subfigure[Evolution of $|\psi|^2$.]{\includegraphics[scale=0.29]{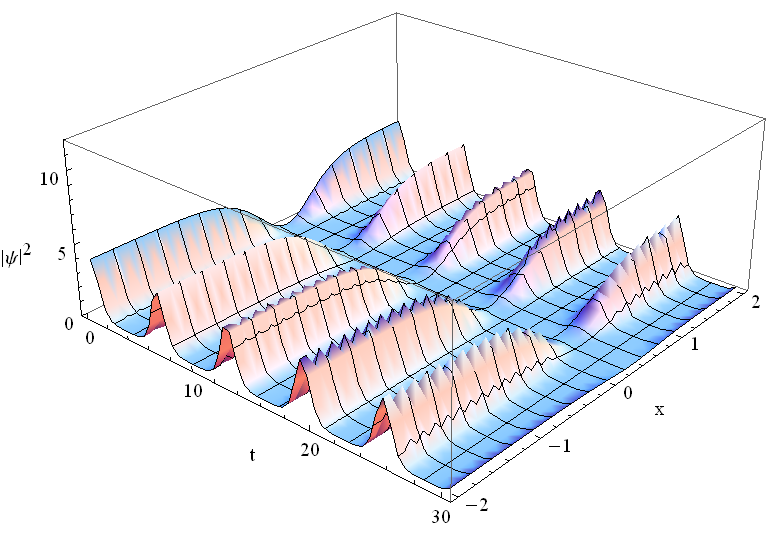}}
\subfigure[Evolution of $|\varphi|^2$.]{\includegraphics[scale=0.31]{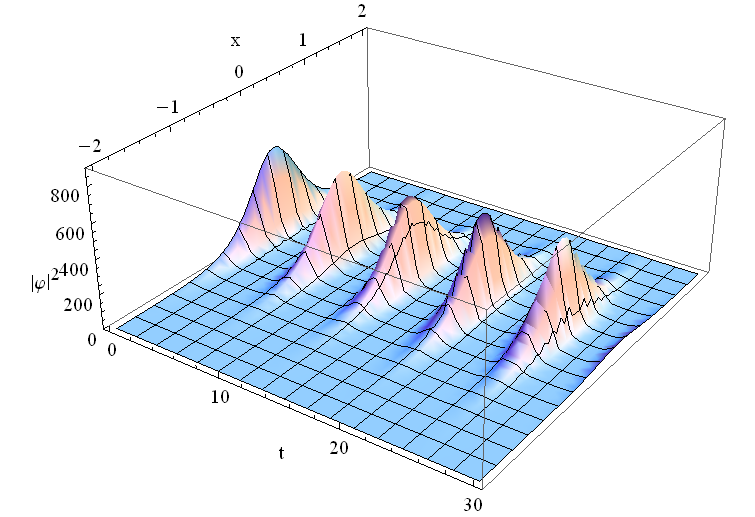}}
\caption{Solutions for the system  (\ref{Ex6a})-(\ref{Ex6b}) for the parameters $e=-1,$ $C=2,$ $a_3 = -1$ and $b_3 = 1$.}\label{Fig6}
\end{figure}

\noindent where $\chi$ and $\phi$ satisfy equations (\ref{Dark_Brighta})-(\ref{Dark_Brightb}), and 
\begin{eqnarray*}
 \alpha(t)= \dfrac{\sin t}{4}, \quad  \quad \quad  \beta(t)=1,\quad  \quad \quad 
 \gamma(t)= t,   \quad  \quad \quad  \delta(t)=1, \\\  
  \varepsilon(t)=-2t, \quad \quad \quad  \kappa(t)=-t,
  \quad \quad \quad  \mu(t)= e^{3-3\cos t}.
\end{eqnarray*}

\noindent In Figure \ref{Fig6}, we display the profiles of  $|\psi|^2,  |\varphi|^2$  by using the values $e=-1,$ $C=2,$ $a_3 = -1$, and $b_3 = 1$.

\subsection{Bending dynamics}
 The ability to manipulate the dynamics of solutions is extremely valuable in practice. In this sense, the Riccati approach used here allows the solutions to be controlled. To demonstrate this point, consider the system (\ref{Ex6a})-(\ref{Ex6b}) again, but this time change the initial conditions for $\delta(t)$ and $\varepsilon(t)$, i.e., $\delta(0) = 0.8,$ $\varepsilon(0) = -1$; and $\delta(0) = 1.5,$ $\varepsilon(0) = -4.5$. The $\delta(t)$, $\varepsilon(t)$, and $\kappa(t)$ functions are now:
\begin{eqnarray*}
  \delta(t)= 0.8, \quad  \quad \quad  
  \varepsilon(t)=-1-1.6t, \quad \quad \quad  \kappa(t)=-0.64t; \\\
   \delta(t)= 1.5, \quad  \quad \quad  
  \varepsilon(t)=-4.5-3t, \quad \quad \quad  \kappa(t)=-2.25t.
\end{eqnarray*}

\begin{figure}[h!]
\centering
\subfigure[Profile of $|\psi|^2$  bended to
the left: $\delta(0) = 0.8$ and $\varepsilon(0) = -1$.]{\includegraphics[scale=0.31]{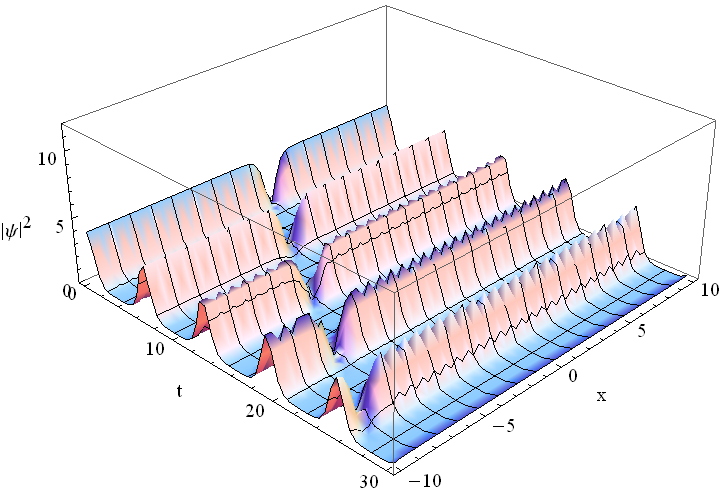}}
\subfigure[Profile of $|\varphi|^2$ bended to the left: $\delta(0) = 0.8$ and $\varepsilon(0) = -1$.]{\includegraphics[scale=0.27]{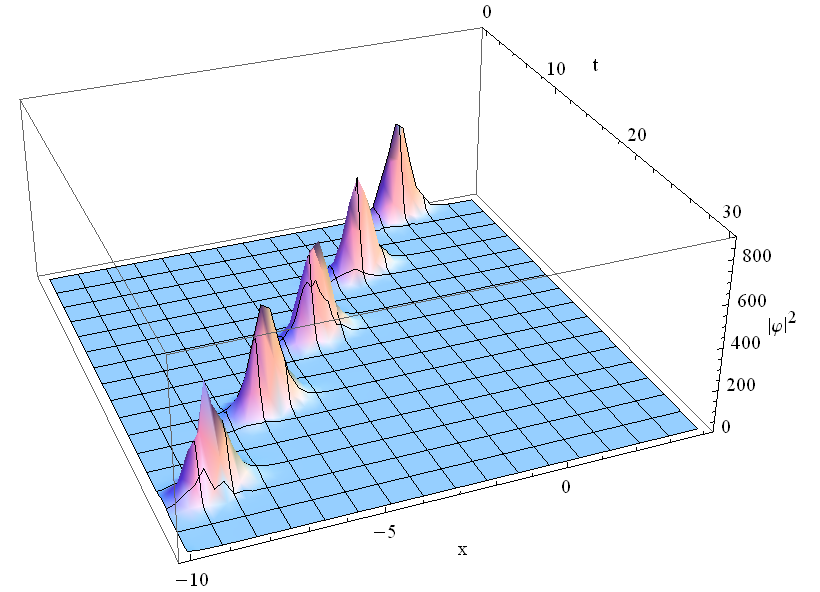}}
\subfigure[Profile of $|\psi|^2$ bended to the right: $\delta(0) = 1.5$ and $\varepsilon(0) = -4.5$.]{\includegraphics[scale=0.29]{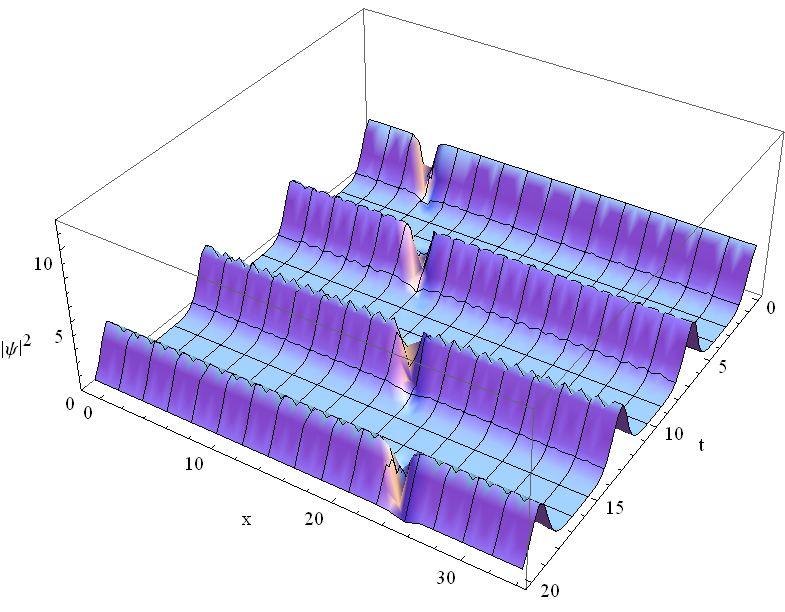}}
\subfigure[Profile of $|\varphi|^2$ bended to the right: $\delta(0) = 1.5$ and $\varepsilon(0) = -4.5$.]{\includegraphics[scale=0.28]{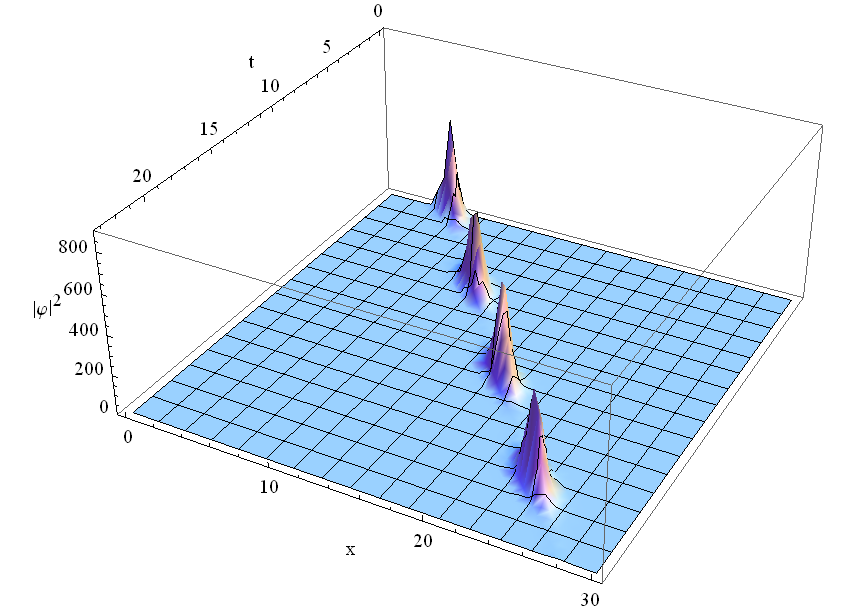}}
\caption{Control of the dynamics of the system's solution (\ref{Ex6a})-(\ref{Ex6b}).}\label{Fig7}
\end{figure}
\noindent Figure \ref{Fig7}, shows how these new parameters influence the system's nonlinear dynamics: The propagation axis of the solutions is bended.



\section{Blowup Solutions of the Generalized Coupled NLS System }

The current section is focused on the study of singular solutions of the coupled VCNLS. More precisely, we demonstrate by explicit construction the existence of finite-time  singular solutions  in the $L^{\infty}$ norm. 
\begin{theorem}[Solutions with singularity in finite time with  $L^{\infty}$ norm]
\label{Theo2}
 Consider the generalized coupled nonlinear Schr\"odinger system
 {\small 
\begin{eqnarray}
i\psi _{t} &=&-a(t)\psi_{xx} +(b\left(
t\right) x^{2}-id(t)-xf(t))\psi +i(g(t)-c\left( t\right) x)\psi
_{x}+h(t)(\left\vert \varphi \right\vert ^{2s}+\left\vert \psi \right\vert
^{2s})\psi,  \label{vcNLS3} \\
i\varphi_{t} &=&-a(t)\varphi_{xx} +(b\left(
t\right) x^{2}-id(t)-xf(t))\varphi +i(g(t)-c\left( t\right) x)\varphi
_{x}+h(t)(\left\vert \varphi \right\vert ^{2s}+\left\vert \psi \right\vert
^{2s})\varphi,  \label{vcNLS4}
\end{eqnarray}%
}
where $a(t),b(t),c(t),d(t),f(t),g(t),h(t)$ are time-dependent functions and $s>0.$ Then, the system (\ref{vcNLS3})-(\ref{vcNLS4}) admits the  $L^{\infty}$ singular solutions 
\begin{equation}
\psi(x,t)=\dfrac{1}{\sqrt{\mu(t)}}\exp\left[i\left(\alpha(t)x^2+\beta(t)xy+\gamma(t)y^2 +\delta(t)x+\varepsilon(t)y+\kappa(t)\right)\right], \label{SolExp1}
\end{equation}
\begin{equation}
\varphi(x,t)=\dfrac{1}{\sqrt{\mu(t)}}\exp\left[i\left(\alpha(t)x^2+\beta(t)xz+\gamma(t)z^2 +\delta(t)x+\varepsilon(t)z+\kappa(t)\right)\right], \label{SolExp2}
\end{equation}
where $y,z$ are real parameters and $\alpha(t),\beta(t),\gamma(t),\delta(t),\varepsilon(t),\kappa(t),\mu(t)$ satisfy the modified Riccati system (\ref{Rica1})-(\ref{Rica6}).
 
\end{theorem}
\begin{proof}
 We will proceed as in \cite{CorderoSoto2008,Escorcia}. Substituting (\ref{SolExp1})-(\ref{SolExp2}) into equations (\ref{vcNLS3})-(\ref{vcNLS4}) we get the modified Riccati system:
 \begin{equation}
\dfrac{d\alpha }{dt}+b(t)+2c(t)\alpha +4a(t)\alpha^{2}=0,  \label{Rica1}
\end{equation}%
\begin{equation}
\dfrac{d\beta }{dt}+(c(t)+4a(t)\alpha(t))\beta =0,  \label{Rica2}
\end{equation}%
\begin{equation}
\dfrac{d\gamma }{dt}+a(t)\beta^{2}(t)=0,
\label{Rica3}
\end{equation}%
\begin{equation}
\dfrac{d\delta }{dt}+(c(t)+4a(t)\alpha(t))\delta = f(t)+2\alpha (t)g(t),
\label{Rica4}
\end{equation}%
\begin{equation}
\dfrac{d\varepsilon }{dt}=(g(t)-2a(t)\delta(t))\beta (t),  \label{Rica5}
\end{equation}%
\begin{equation}
\dfrac{d\kappa }{dt} = g(t)\delta (t)-a(t)\delta^{2}(t)-\dfrac{2h(t)}{\mu^{s}(t)}.  \label{Rica6}
\end{equation}
 Here the previous system (\ref{Rica1})-(\ref{Rica6})   differs from (\ref{rica1})-(\ref{rica6}) by the extra term   $-2h(t)/\mu^{s}(t)$ in equation (\ref{Rica6}). Now, considering the substitution 
 \begin{equation}
\alpha = \left(\dfrac{1}{4a(t)}\dfrac{\mu ^{\prime }(t)}{\mu (t)}-\dfrac{d(t)}{2a(t)%
}\right),  \label{Sus1}
\end{equation}%
it follows that the Riccati equation (\ref{Rica1}) becomes\ 
\begin{equation}
\mu ^{\prime \prime }-\eta(t)\mu ^{\prime }+4\sigma (t)\mu =0,
\label{carac1}
\end{equation}%
with\ 
\begin{equation}
\eta (t)=\frac{a^{\prime }}{a}-2c+4d,\hspace{1cm}\sigma (t)=ab-cd+d^{2}+%
\frac{d}{2}\left( \frac{a^{\prime }}{a}-\frac{d^{\prime }}{d}\right) .
\end{equation}
 Using Equations (\ref{mu})-(\ref{epsilon0}) in the Appendix with $l_0 = 1$, (\ref{Rica1})-(\ref{Rica5}) can be solved, but (\ref{Rica6})  must be solved separately.
 Let $I$ be  an interval such that $\mu_{0}(t) \not = 0$ for all $t\in I.$ Now, since the functions $\mu_{0}$ and $\mu_1$ are linearly independent on an interval $I^\prime$, we have
 
 \begin{equation}
  \gamma_{0}^{\prime}(t) = \frac{W\left[\mu_{0}(t),\mu_{1}(t) \right]}{2\mu_{0}^{2}(t)} \not = 0, \quad t\in I \cap I^{\prime},
 \end{equation}
and by $\mu(t)$ given by (\ref{mu}), see \cite{CorderoSoto2011}, functions (\ref{SolExp1})-(\ref{SolExp2})  will have a  finite time blowup  in the $L^{\infty}$ norm at $T_{b} = \gamma_{0}^{\prime}(-\alpha(0))\in I\cap I^{\prime}$, i.e., 
\begin{equation}
 |\psi(x,t)|, |\varphi(x,t)| \to \infty \quad  \textrm{as} \quad  t \to T_b.
\end{equation}
\end{proof}

\subsubsection{Case $a = 1/2,  \ b = c = d = f = g = 0,  $ and arbitrary $h$}
 Consider the coupled VCNLS system
 {\small 
\begin{eqnarray}
i\psi _{t} &=&-\frac{1}{2}\psi_{xx} +h(t)(\left\vert \varphi \right\vert ^{2s}+\left\vert \psi \right\vert
^{2s})\psi,  \label{example9a} \\
i\varphi_{t} &=&-\frac{1}{2}\varphi_{xx} +h(t)(\left\vert \varphi \right\vert ^{2s}+\left\vert \psi \right\vert
^{2s})\varphi.  \label{example9b}
\end{eqnarray}%
}
Then, according to  Theorem \ref{Theo2}, the system admits the explicit  solutions given by the Equations (\ref{SolExp1})-(\ref{SolExp2}) with the functions
\begin{eqnarray*}
 \alpha(t)= \frac{\alpha(0)}{1+2\alpha(0)t}, \quad \quad \quad    \beta(t)= \frac{\beta(0)}{1+2\alpha(0)t},\quad   \quad \quad  \gamma(t)= \gamma(0)-\frac{\beta^2(0)t}{2+4\alpha(0)t},  \\\
  \delta(t)=\frac{\delta(0)}{1+2\alpha(0)t},  \quad \quad \quad \kappa(t)=\kappa(0)+\frac{\delta^2(0) t}{4\alpha(0)(1+2\alpha(0)t)}-\frac{2}{\mu^{s}(0)}\int_{0}^{t}\frac{h(t^\prime)}{(1+2\alpha(0)t^\prime)^s}dt^\prime, \\\  \varepsilon(t)=\varepsilon(0)-\frac{\beta(0)\delta(0)t}{1+2\alpha(0)t}, \quad \quad \quad \quad \quad  \quad  \mu(t)= \mu(0)(1+2\alpha(0)t).
\end{eqnarray*}
Therefore, these solutions blow up in the $L^{\infty}$ norm at the instant $T_{b}= -\frac{1}{2\alpha(0)}$, provided $\alpha(0)<0.$ 

In this context, this last result allows us to extend the singular solutions established in Theorem 1 in \cite{Escorcia} (see also Example 1 in the same reference).



\section{An n-dimensional Generalized Coupled NLS System}
In Section 3, we established a relationship between the one-dimensional coupled VCNLS (\ref{vcNLS1})-(\ref{vcNLS2}) and the well-known Manakov equations (\ref{ManakovSystem1})-(\ref{ManakovSystem2}), provided the coefficients satisfy the Riccati system (\ref{rica1})-(\ref{rica6}).  In the subsequent Section 4 and by ``exploiting'' the potential of the Riccati system,  the construction of explicit singular solutions for (\ref{vcNLS1})-(\ref{vcNLS2}) was shown. In this section, we reveal that a similar Riccati system (see Equations (\ref{RM1})-(\ref{RM6})) allows us to link an n-dimensional coupled NLS system with a generalized Manakov system. It will be useful in constructing soliton  solutions for the high-dimensional system.

\begin{theorem}
The nonautonomous coupled NLS n-dimensional
\begin{equation}
i\psi _{t}=-a(t)\Delta \psi +b(t)|\boldsymbol{x}|^{2}\psi +\lambda a(t)\beta^{2}(t)\mu
^{s}(t)(|\varphi |^{2s}+|\psi |^{2s})\psi , \label{nDa}
\end{equation}%
\begin{equation}
i\varphi _{t}=-a(t)\Delta \varphi +b(t)|\boldsymbol{x}|^{2}\varphi + \lambda a(t)\beta
^{2}(t)\mu ^{s}(t)(|\varphi |^{2s}+|\psi |^{2s})\varphi , \label{nDb}
\end{equation}%
with  $t\in \mathbb{R}$   and $\boldsymbol{x}\in \mathbb{R}^n$ can be transformed into the standard coupled NLS%
\begin{equation}
iu_{\tau }= \Delta_{\boldsymbol{\xi}} u - \lambda (|v|^{2s}+|u|^{2s})u,
\end{equation}%
\begin{equation}
iv_{\tau }= \Delta_{\boldsymbol{\xi}} v-\lambda (\left\vert v\right\vert
^{2s}+|u|^{2s})v,
\end{equation}%
by the transformations\ 
\begin{equation}
\psi (\boldsymbol{x},t)=\dfrac{1}{\sqrt{\mu (t)}}e^{i(\sum_{i=1}^{n}\alpha
(t)x_{i}^{2}+\delta _{i}(t)x_{i}+\kappa _{i}(t))}u(\boldsymbol{\xi },\tau), \label{GNa}
\end{equation}%
\ 
\begin{equation}
\varphi (\boldsymbol{x},t)=\dfrac{1}{\sqrt{\mu (t)}}e^{i(\sum_{i=1}^{n}%
\alpha (t)x_{i}^{2}+\delta _{i}(t)x_{i}+\kappa _{i}(t))}v(\boldsymbol{%
\xi },\tau), \label{GNb} \qquad
\end{equation}
with $\xi _{i}=\beta (t)x_{i}+\varepsilon _{i}(t)$ for $i = 1,2,..., n$ and $\tau =\gamma (t).$ Additionally, the following Riccati system is satisfied 
\begin{equation}
\label{RM1}
\dfrac{d\alpha }{dt}+b(t)+4a(t)\alpha ^{2}= 0,
\end{equation}%
\begin{equation}
\dfrac{d\beta }{dt}+4a(t)\alpha (t)\beta =0,
\end{equation}%
\begin{equation}
\dfrac{d\gamma }{dt}+ a(t)\beta ^{2}(t)=0,
\end{equation}%
\begin{equation}
\dfrac{d\delta _{i}}{dt}+4a(t)\alpha (t)\delta _{i}= 0,
\end{equation}%
\begin{equation}
\dfrac{d\varepsilon _{i}}{dt} + 2a(t)\beta (t)\delta _{i}(t) = 0,
\end{equation}%
\begin{equation}
\dfrac{d\kappa _{i}}{dt} + a(t)\delta _{i}^{2}(t) = 0,
\label{RM6}
\end{equation}%
\ $i=1,...,n.$ Considering the  substitution\ 
\begin{equation}
\alpha =\dfrac{\mu ^{\prime }(t)}{4 n a(t) \mu (t)},
\end{equation}%
it follows that the Riccati equation (\ref{RM1}) becomes\ 
\begin{equation}
\mu ^{\prime \prime }-\frac{a^\prime}{a}\mu ^{\prime }+4nab\mu +\left(\frac{1-n}{n}\right)\frac{(\mu^{\prime})^2}{\mu} =0.
\end{equation}%
\end{theorem}  
In general, the study of n-dimensional coupled NLS systems is carried out by degenerating them into the one-dimensional coupled NLS, see for example \cite{Lan,Kannan} for the two-dimensional case, or directly dealing with  the high-dimensional system  as can be seen in \cite{Zhong-ZhouLan} for $n = 3$. In contrast, our transformation does not reduce the spatial dimension of the system, and then, to find soliton solutions for (\ref{nDa})-(\ref{nDb}) in terms of Proposition 1, we  used the ideas given in the works  \cite{Zhong-ZhouLan,Lan,Kannan}.

\subsubsection{Case $a = 2,  \ b = -2, \ h = -8$}
 The two-dimensional coupled NLS system
 {\small 
\begin{eqnarray}
i\psi _{t} &=&-2(\psi_{xx}+\psi_{yy}) -2(x^2 + y^2)\psi -8(\left\vert \varphi \right\vert ^{2}+\left\vert \psi \right\vert ^{2})\psi,  \label{example10a} \\
i\varphi_{t} &=&-2(\varphi_{xx}+\varphi_{yy}) -2(x^2 + y^2)\varphi -8(\left\vert \varphi \right\vert ^{2}+\left\vert \psi \right\vert ^{2})\varphi  \label{example10b}
\end{eqnarray}%
}
admits the DB soliton solution given by equations (\ref{GNa})-(\ref{GNb}) with
\begin{eqnarray*}
 \alpha(t)= \frac{1}{2}\tanh(4t) , \quad \quad \quad \quad  \quad   \beta(t)= \sech(4t) ,\quad   \quad \quad \quad  \gamma(t)= -\frac{1}{2}\tanh(4t) ,  \\\
  \delta_{i}(t)= \sech(4t),  \quad \quad \kappa_{i}(t)= -\frac{1}{2}\tanh(4t),  \quad \quad  \varepsilon_{i}(t) = -\tanh(4t) , \quad \quad  \mu(t)= \cosh^{2}(4t),
\end{eqnarray*}
and the functions \cite{Lan,Kannan}:
\begin{eqnarray}
    u(\xi_1,\xi_2,\tau) = \chi(\xi_1 + \xi_2, -2\tau),
\end{eqnarray}
\begin{eqnarray}
    v(\xi_1,\xi_2,\tau) = \phi(\xi_1 + \xi_2, -2\tau),
\end{eqnarray}

with $\chi$ and $\phi$ given by the Equations (\ref{Dark_Brighta})-(\ref{Dark_Brightb}).

\subsubsection{Case $a = 2,  \ b = -\frac{1+2t^2}{4}, \ h = -12e^{t^2}$}
 Consider the three-dimensional coupled NLS system
 {\small 
\begin{eqnarray}
i\psi _{t} &=&-2(\psi_{xx}+\psi_{yy} + \psi_{zz}) -\frac{1+2t^2}{4}(x^2 + y^2 + z^2)\psi -12e^{t^2}(\left\vert \varphi \right\vert ^{2}+\left\vert \psi \right\vert ^{2})\psi,  \label{example11a} \\
i\varphi_{t} &=&-2(\varphi_{xx}+\varphi_{yy} + \varphi_{zz}) -\frac{1+2t^2}{4}(x^2 + y^2 + z^2)\varphi -12e^{t^2}(\left\vert \varphi \right\vert ^{2}+\left\vert \psi \right\vert ^{2})\varphi.  \label{example11b}
\end{eqnarray}%
}

\noindent Then, such a system admits the type I Rogue wave solution given by the Equations (\ref{GNa})-(\ref{GNb}) with
\begin{eqnarray*}
 \alpha(t)= \frac{t}{4}, \quad \quad \quad \quad  \quad   \beta(t)= e^{-t^2} ,\quad   \quad \quad \quad  \gamma(t)= -\sqrt{\frac{\pi}{2}}\mbox{Erf}(\sqrt{2}t),  \\\
  \delta_{i}(t)= e^{-t^2},  \quad \quad \kappa_{i}(t)= -\sqrt{\frac{\pi}{2}}\mbox{Erf}(\sqrt{2}t),  \quad \quad  \varepsilon_{i}(t) = -\sqrt{2 \pi}\mbox{Erf}(\sqrt{2}t), \quad \quad  \mu(t)= e^{3t^2},
\end{eqnarray*}
and the functions
\begin{eqnarray}
    u(\xi_1,\xi_2,\xi_3,\tau) = \chi(\xi_1 + \xi_2 + \xi_3, -3\tau), \quad \quad 
    v(\xi_1,\xi_2,\xi_3,\tau) = \phi(\xi_1 + \xi_2 + \xi_3, -3\tau),
\end{eqnarray}
with $\chi$ and $\phi$ given by equations (\ref{RWIa})-(\ref{RWIb}).

\section{Conclusions and Final Remarks}
Through this current work, we proposed and showed the integrability of a generalized coupled  VCNLS  using the similarity transformation established in Theorem \ref{Th1} and Theorem 3. Here, the Riccati system played a crucial role in the explicit construction of soliton-type solutions, rogue waves, and singular solutions for these kinds of equations. The explicit construction of soliton and Rogue wave-type solutions for an n-dimensional coupled NLS system was shown as well. This work should motivate further analytical and numerical studies looking to clarify the connections between  the dynamics of
variable-coefficient NLS and the dynamics of Riccati systems. We're hoping this fresh approach to dealing with VCNLS systems will provide new insights into the physics of such models and aid in comprehending real-world applications.

\begin{acknowledgement}
We thank the anonymous reviewers for their valuable time spent reading our work.
\end{acknowledgement}

\section{Appendix: Solution of the Riccati System }

In this appendix, we provide a solution of the Riccati system (\ref{rica1})-(\ref{rica6}), and we point out that all the formulas involved in such a solution have been verified previously in \cite{Ko-su-su}. A solution of the Riccati system with multiparameters is given by the following expressions (with the respective
inclusion of the parameter $l_{0}$) \cite{CorderoSoto2008,Escorcia,Suazo009,Suslov12}:

\begin{equation}
\mu \left( t\right) =2\mu \left( 0\right) \mu _{0}\left( t\right) \left(
\alpha \left( 0\right) +\gamma _{0}\left( t\right) \right) ,  \label{mu}
\end{equation}%
\begin{equation}
\alpha \left( t\right) =\alpha _{0}\left( t\right) -\frac{\beta
_{0}^{2}\left( t\right) }{4\left( \alpha \left( 0\right) +\gamma _{0}\left(
t\right) \right) },  \label{alpha}
\end{equation}%
\begin{equation}
\beta \left( t\right) =-\frac{\beta \left( 0\right) \beta _{0}\left(
t\right) }{2\left( \alpha \left( 0\right) +\gamma _{0}\left( t\right)
\right) }=\frac{\beta \left( 0\right) \mu \left( 0\right) }{\mu \left(
t\right) }w\left( t\right) ,  \label{beta}
\end{equation}%
\begin{equation}
\gamma \left( t\right) =l_{0}\gamma \left( 0\right) -\frac{l_{0}\beta
^{2}\left( 0\right) }{4\left( \alpha \left( 0\right) +\gamma _{0}\left(
t\right) \right) },  \label{gamma}
\end{equation}%
\begin{equation}
\delta \left( t\right) =\delta _{0}\left( t\right) -\frac{\beta _{0}\left(
t\right) \left( \delta \left( 0\right) +\varepsilon _{0}\left( t\right)
\right) }{2\left( \alpha \left( 0\right) +\gamma _{0}\left( t\right) \right) 
},  \label{delta}
\end{equation}%
\begin{equation}
\varepsilon \left( t\right) =\varepsilon \left( 0\right) -\frac{\beta \left(
0\right) \left( \delta \left( 0\right) +\varepsilon _{0}\left( t\right)
\right) }{2\left( \alpha \left( 0\right) +\gamma _{0}\left( t\right) \right) 
},  \label{epsilon}
\end{equation}%
\begin{equation}
\kappa \left( t\right) =\kappa \left( 0\right) +\kappa _{0}\left( t\right) -%
\frac{\left( \delta \left( 0\right) +\varepsilon _{0}\left( t\right) \right)
^{2}}{4\left( \alpha \left( 0\right) +\gamma _{0}\left( t\right) \right) },
\label{kappa}
\end{equation}%
\ subject to the initial arbitrary conditions $\mu \left( 0\right) ,$ $%
\alpha \left( 0\right) ,$ $\beta \left( 0\right) \neq 0,$ $\gamma (0),$ $%
\delta (0),$ $\varepsilon (0)$ and $\kappa (0)$. $\alpha _{0}$, $\beta _{0}$%
, $\gamma _{0}$, $\delta _{0}$, $\varepsilon _{0}$ and $\kappa _{0}$ are
given explicitly by\ 
\begin{equation}
\alpha _{0}\left( t\right) =\frac{1}{4a\left( t\right) }\frac{\mu
_{0}^{\prime }\left( t\right) }{\mu _{0}\left( t\right) }-\frac{d\left(
t\right) }{2a\left( t\right) },  \label{alpha0}
\end{equation}%
\begin{equation}
\beta _{0}\left( t\right) =-\frac{w\left( t\right) }{\mu _{0}\left( t\right) 
},\quad w\left( t\right) =\exp \left( -\int_{0}^{t}\left( c\left( s\right)
-2d\left( s\right) \right) \ ds\right) ,  \label{beta0}
\end{equation}%
\begin{equation}
\gamma _{0}\left( t\right) =\frac{1}{2\mu _{1}\left( 0\right) }\frac{\mu
_{1}\left( t\right) }{\mu _{0}\left( t\right) }+\frac{d\left( 0\right) }{%
2a\left( 0\right) },  \label{gamma0}
\end{equation}%
\begin{equation}
\delta _{0}\left( t\right) =\frac{w\left( t\right) }{\mu _{0}\left( t\right) 
}\ \ \int_{0}^{t}\left[ \left( f\left( s\right) -\frac{d\left( s\right) }{%
a\left( s\right) }g\left( s\right) \right) \mu _{0}\left( s\right) +\frac{%
g\left( s\right) }{2a\left( s\right) }\mu _{0}^{\prime }\left( s\right) %
\right] \ \frac{ds}{w\left( s\right) },  \label{delta0}
\end{equation}%
\begin{eqnarray}
\varepsilon _{0}\left( t\right) &=&-\frac{2a\left( t\right) w\left( t\right) 
}{\mu _{0}^{\prime }\left( t\right) }\delta _{0}\left( t\right)
+8\int_{0}^{t}\frac{a\left( s\right) \sigma \left( s\right) w\left( s\right) 
}{\left( \mu _{0}^{\prime }\left( s\right) \right) ^{2}}\left( \mu
_{0}\left( s\right) \delta _{0}\left( s\right) \right) \ ds  \label{epsilon0}
\\
&&+2\int_{0}^{t}\frac{a\left( s\right) w\left( s\right) }{\mu _{0}^{\prime
}\left( s\right) }\left[ f\left( s\right) -\frac{d\left( s\right) }{a\left(
s\right) }g\left( s\right) \right] \ ds,  \notag
\end{eqnarray}%
\begin{eqnarray}
\kappa _{0}\left( t\right) &=&\frac{a\left( t\right) \mu _{0}\left( t\right) 
}{\mu _{0}^{\prime }\left( t\right) }\delta _{0}^{2}\left( t\right)
-4\int_{0}^{t}\frac{a\left( s\right) \sigma \left( s\right) }{\left( \mu
_{0}^{\prime }\left( s\right) \right) ^{2}}\left( \mu _{0}\left( s\right)
\delta _{0}\left( s\right) \right) ^{2}\ ds  \label{kappa0} \\
&&\quad -2\int_{0}^{t}\frac{a\left( s\right) }{\mu _{0}^{\prime }\left(
s\right) }\left( \mu _{0}\left( s\right) \delta _{0}\left( s\right) \right) %
\left[ f\left( s\right) -\frac{d\left( s\right) }{a\left( s\right) }g\left(
s\right) \right] \ ds,  \notag
\end{eqnarray}%
\ with $\delta _{0}\left( 0\right) =g_{0}\left( 0\right) /\left( 2a\left(
0\right) \right) ,$ $\varepsilon _{0}\left( 0\right) =-\delta _{0}\left(
0\right) ,$ $\kappa _{0}\left( 0\right) =0.$ Here $\mu _{0}$ and $\mu _{1}$
represent the fundamental solution of the characteristic equation subject to
the initial conditions $\mu _{0}(0)=0$, $\mu _{0}^{\prime }(0)=2a(0)\neq 0$
and $\mu _{1}(0)\neq 0$, $\mu _{1}^{\prime }(0)=0$.




\begin{thebibliography}{99}
\bibitem{Moll03} K.D.~Moll and A.L.~Gaeta and G. Fibich, Self-similar optical wave collapse: Observation of the Townes profile, Phys. Rev. Lett.,90 No. 20, 2003.
\bibitem{Weinstein83} M.I.~Weinstein, Nonlinear {Schr{\"o}dinger} equations and sharp interpolation estimates ,  Comm. Math. Phys. , 87, 567-576, 1983.

\bibitem{Acosta-Humanez} Primitivo Acosta-Humanez and Erwin Suazo. Liouvillian propagators, Riccati equation and differential Galois
theory. J. Phys. A: Math. Theor., 46:455203, 2013.

\bibitem{Akhmediev} N. Akhmediev, A. Ankiewicz, and J.M. Soto-Crespo. Rogue waves and rational solutions of the nonlinear
Schrodinger equation. Phys. Rev. E, 80:026601, 2009.

\bibitem{Suazo16}Gabriel Amador, Kiara Colon, Nathalie Luna, Gerardo Mercado, Enrique Pereira, and Erwin Suazo. On solutions
for linear and nonlinear Schrodinger equations with variable coefficients: A computational approach. Symmetry,
8 (39), 2016.

\bibitem{BoLi} G. Bo-Ling and L. Li-Ming. Rogue wave, Breathers and Bright-Dark-Rogue solutions for the coupled Schrodinger
equations. Chin. Phys. Lett., 28, No 11:110202, 2011.

\bibitem{Busch2001}Th. Busch and J. R. Anglin. Dark-Bright solitons in inhomogeneous Bose-Einstein condensates. Phys. Rev. Lett.,
87:010401, 2001.

\bibitem{Chai} Jun Chai, Bo Tian, and Han-Peng Chai. Darboux transformation and vector solitons for a variable-coefficient
coherently coupled nonlinear Schrodinger system in nonlinear optics. Optical Engineering, 55 (11):116113, 2016.

\bibitem{Chai15}Jun Chai, Bo Tian, Hui-Ling Zhen, Wen-Rong Sun, and De-Yin Liu. Bright and dark solitons and Backlund
transformations for the coupled cubic-quintic nonlinear Schrodinger equations with variable coefficients in an
optical fiber. Phys. Scr., 90:045206, 2015.

\bibitem{Chakraborty} Sushmita Chakraborty, Sudipta Nandy, and Abhijit Barthakur. Bilinearization of the generalized coupled non-
linear Schrodinger equation with variable coefficients and gain and dark-bright pair soliton solutions. Physical
Review E, 91:023210, 2015.

\bibitem{Zhigang1997}  Z. Chen, M. Segev, T.H. Coskun, D.N. Christodoulides, and Y.S. Kivshar. Coupled photorefractive spatial-
soliton pairs. J. Opt. Soc. Am. B, 14, No. 11, 1997.

\bibitem{CorderoSoto2008} R. Cordero-Soto, R.M. Lopez, E. Suazo, and S.K. Suslov. Propagator of a charged particle with a spin in uniform magnetic and perpendicular electric fields. Lett. Math. Phys., 84 (2-3):159–178, 2008.

\bibitem{CorderoSoto2011} R. Cordero-Soto and S.K. Suslov. The degenerate parametric oscillator and Ince’s equation. J. Phys. A: Math Theor., 44:015101, 2011.

\bibitem{Du} Zhong Du, Bo Tian, Han-Peng Chai, Yan Sun, and Xue-Hui Zhao. Rogue waves for the coupled variable-
coefficient fourth-order nonlinear Schrodinger equations in an inhomogeneous optical fiber. Chaos, Solitons and
Fractals, 109:90–98, 2018.

\bibitem{Rehab} Rehab M. El-Shiekh and Mahmoud Gaballah. Solitary wave solutions for the variable-coefficient coupled nonlin-
ear Schrodinger equations and Davey–Stewartson system using modified sine-Gordon equation method. Journal
of Ocean Engineering and Science, 5:180–185, 2020.

\bibitem{Escorcia} J. Escorcia and E. Suazo. Blow-up results and soliton solutions for a generalized variable coefficient nonlinear
Schrodinger equation. Applied Mathematics and Computation, 301:155–176, 2017.

\bibitem{Han} Lijia Han, Yehui Huang, and Hui Liu. Solitons in coupled nonlinear Schrodinger equations with variable coeffi-
cients. Commun Nonlinear Sci Numer Simulat, 19:3063–3073, 2014.

\bibitem{Ho} Tin-Lun Ho and V. B. Shenoy. Binary mixtures of Bose condensates of alkali atoms. Physical Review Letters,
77 (16):3276–3279, 1996.

\bibitem{Kevrekidis} P. G. Kevrekidis and D. J. Frantzeskakis. Solitons in coupled nonlinear Schrodinger models: A survey of recent
developments. Reviews in Physics, 1:140–153, 2016.

\bibitem{Kevrekidis08} P.G. Kevrekidis, D.J. Frantzeskakis, and R. Carretero-Gonzalez. Emergent nonlinear phenomena in Bose-
Einstein condensates: Theory and experiment. Springer, Berlin, 2008.

\bibitem{Ko-su-su} Christoph Koutschan, Erwin Suazo, and Sergei K. Suslov. Fundamental laser modes in paraxial optics: from
computer algebra and simulations to experimental observation. Appl. Phys. B, 2015.

\bibitem{Zhong-ZhouLan} Zhong-Zhou Lan. Dark solitonic interactions for the (3 + 1)-dimensional coupled nonlinear Schrodinger equations in nonlinear optical fibers. Optics and Laser Technology, 113:462–466, 2019.

\bibitem{Lan} Zhong-Zhou Lan, Bo Gao, and Ming-Jing Du. Dark solitons behaviors for a (2+1)-dimensional coupled nonlinear
Schrodinger system in an optical fiber. Chaos, Solitons and Fractals, 111:169–174, 2018.

\bibitem{Xiaoyan} Xiaoyan Liu, Qin Zhou, Anjan Biswas, Abdullah Kamis Alzahrani, and Wenjun Liu. The similarities and
differences of different plane solitons controlled by (3 + 1)-dimensional coupled variable coefficient system.
Journal of Advanced Research, 24:167–173, 2020.

\bibitem{Manakov1974}  S. V. Manakov. On the theory of two-dimensional stationary self-focusing of electromagnetic waves. Sov. Phys. JETP, 38, No. 2, 1974.

\bibitem{Manganaro} N. Manganaro and D. F. Parker. Similarity reductions for variable-coefficient coupled nonlinear Schrodinger
equations. J. Phys. A: Math. Gen., 26:4093–4106, 1993.

\bibitem{Kannan} Kannan Manikandan, Murugaian Senthilvelan, and Roberto Andre  Kraenkel. On the characterization of vector
rogue waves in two-dimensional two coupled nonlinear Schrodinger equations with distributed coefficients. Eur.
Phys. J. B, 89 (218), 2016.

\bibitem{Menyuk} C.R. Menyuk. Stability of solitons in birefringent optical fibers. Optical fibers, J. Opt. Soc. Am. B, 5:392–402,
1998.

\bibitem{Ostrovskaya1999} E.A. Ostrovskaya, Y.S. Kivshar, Z. Chen, and M. Segev. Interaction between vector solitons and solitonic gluons. Opt. Lett., 24:327–329, 1999.

\bibitem{Muslum} Muslum Ozisik, Aydin Secer, and Mustafa Bayram. On the examination of optical soliton pulses of Manakov
system with auxiliary equation technique. Optik-International Journal for light and electron optics, 268, 2022.

\bibitem{Arash} Arash Pashrashid, Cesar A. Gomez S, Seyed M. Mirhosseini-Alizamini, Seyed Navid Motevalian, M. Daher
Albalwi, Hijaz Ahmad, and Shao-Wen Yao. On travelling wave solutions to Manakov model with variable
coefficients. Open Physics, 2023.

\bibitem{Pereira} Enrique Pereira, Erwin Suazo, and Jessica Trespalacios. Riccati–Ermakov systems and explicit solutions for
variable coefficient reaction–diffusion equations. Applied Mathematics and Computation, 329:278–296, 2018.

\bibitem{Juncai} Jun-Cai Pu and Yong Chen. Data-driven vector localized waves and parameters discovery for Manakov system
using deep learning approach. Chaos, Solitons and Fractals, 160, 2022.

\bibitem{Qiu} Yuting Qiu and Ping Gao. New exact solutions for the coupled nonlinear Schrodinger equations with variable
coefficients. Journal of Applied Mathematics and Physics, pages 1515–1523, 2020.

\bibitem{Raissi} M. Raissi, P. Perdikaris, and G.E. Karniadakis. Physics-informed neural networks: A deep learning framework for
solving forward and inverse problems involving nonlinear partial differential equations. Journal of Computational
Physics, 378:686–707, 2019.

\bibitem{ Roskes} G.J. Roskes. Some nonlinear multiphase interactions. Stud. Appl. Math., 55, 1976.

\bibitem{Shi} Z. Shi and J. Yang. Solitary waves bifurcated from Bloch-band edges in two-dimensional periodic media. Phys.
Rev. E, 75:056602, 2007.

\bibitem{Suazo009} E. Suazo. Fundamental solutions of some evolution equations. Arizona State University, Sep. 2009 (Phd. disser-
tation).

\bibitem{Suazo18} E. Suazo and S.K. Suslov. An integral form of the nonlinear Schrodinger equation with variable coefficients.
Progress in Electromagnetics Research Symposium (PIERS-Toyama), pages 1214–1220, 2018.

\bibitem{Suazo} Erwin Suazo and Sergei K. Suslov. Soliton-like solutions for the nonlinear Schrodinger equation with variable
quadratic Hamiltonians. Journal of Russian Laser Research, 33 (1), 2012.

\bibitem{Suazo13} Erwin Suazo, Sergei K. Suslov, and Jose M. Vega-Guzman. The Riccati system and a diffusion-type equation.
Mathematics, 2:96–118, 2014.

\bibitem{Suslov12} Sergei K. Suslov. On integrability of nonautonomous nonlinear Schrodinger equations. Am. Math. Soc., 140
(9):3067–3082, 2012.

\bibitem{Timmermans} E. Timmermans. Phase separation of Bose-Einstein condensates. Physical Review Letters, 81 (26):5718–5721,
1998.

\bibitem{Yilmaz} Esra Unal Yilmaz, Farid Samsami Khodad, Yesim Saglam Ozkan, Reza Abazari, A.E. Abouelregal,
Mayssam Tarighi Shaayesteh, Hadi Rezazadeh, and Hijaz Ahmad. Manakov model of coupled NLS equation
and its optical soliton solutions. Journal of Ocean Engineering and Science, 2022.

\bibitem{Yu} Fajun Yu and Zhenya Yan. New rogue waves and dark-bright soliton solutions for a coupled nonlinear Schrodinger
equation with variable coefficients. Applied Mathematics and Computation, 233:351–358, 2014.

\bibitem{ZAKHAROV1982} V. E. Zakharov and E. I. Schulman. To the integrability of the system of two coupled nonlinear Schrodinger
equations. Physica D, 4:270–274, 1982.

\bibitem{Zhang} Ling-Ling Zhang and Xiao-Min Wang. Bright–dark soliton dynamics and interaction for the variable coefficient
three-coupled nonlinear Schrodinger equations. Modern Physics Letters B, page 2050064, 2019







\end{thebibliography}
\end{document}